\documentclass[journal]{IEEEtran}
\usepackage{cite}
\usepackage{amsmath,amssymb,amsfonts}
\usepackage{algorithm} 
\usepackage{algorithmicx} 
\usepackage{graphicx}
\usepackage{textcomp}
\usepackage{makecell}
\usepackage{algpseudocode} 
\usepackage{subcaption}
\usepackage[export]{adjustbox} 
\usepackage{bm} 
\usepackage{booktabs} 
\usepackage{listings}
\usepackage{hyperref}
\usepackage{braket}
\usepackage{multirow}
\usepackage{amsthm}
\newtheorem{theorem}{Theorem}
\newtheorem{lemma}{Lemma}
\newtheorem{definition}{Definition}

\usepackage{qcircuit}
\usepackage{environ}
\newcommand{\myqctmp}[2][0.25]{\Qcircuit @C=#2em @R=#1em @!R}
\NewEnviron{myqcircuit}[1][0.25]{\vcenter{\myqctmp[#1]{0.75} {\BODY}}}
\NewEnviron{myqcircuit*}[2]{\vcenter{\myqctmp[#1]{#2} {\BODY}}}
\newcommand{\controlsq}{*!<0em,.025em>-=-<0em>{\square}}
\newcommand{\ctrlsq}[1]{\controlsq \qwx[#1] \qw}

\def\BibTeX{{\rm B\kern-.05em{\sc i\kern-.025em b}\kern-.08em
    T\kern-.1667em\lower.7ex\hbox{E}\kern-.125emX}}

\begin{document}

\title{Binary Tree Block Encoding of Classical Matrix}

\author{Zexian Li, Xiao-Ming Zhang, Chunlin Yang and Guofeng Zhang\thanks{Zexian Li and Guofeng Zhang are with Department of Applied Mathematics, The Hong Kong Polytechnic University, Hong Kong, China and also with Shenzhen Research Institute, The Hong Kong Polytechnic University, Shenzhen, China (e-mail: zexian.li@connect.polyu.hk; guofeng.zhang@polyu.edu.hk). Xiao-Ming Zhang is with Key Laboratory of Atomic and Subatomic Structure and Quantum Control (Ministry of Education),
South China Normal University, Guangzhou, China, and
Guangdong Provincial Key Laboratory of Quantum Engineering and Quantum Materials, 
Guangdong-Hong Kong Joint Laboratory of Quantum Matter,
South China Normal University, Guangzhou, China (e-mail: phyxmz@gmail.com). Chunlin Yang is with Harbin Engineering University, Heilongjiang Province, China.},~\IEEEmembership{Senior Member,~IEEE} }

\maketitle

\begin{abstract}
Block-encoding is a critical subroutine in quantum computing, enabling the transformation of classical data into a matrix representation within a quantum circuit. The resource trade-offs in simulating a block-encoding can be quantified by the circuit size, the normalization factor, and the time and space complexity of parameter computation. Previous studies have primarily focused either on the time and memory complexity of computing the parameters, or on the circuit size and normalization factor in isolation, often neglecting the balance between these trade-offs. In early fault-tolerant quantum computers, the number of qubits is limited. For a classical matrix of size $2^{n}\times 2^{n}$, our approach not only improves the time of decoupling unitary for block-encoding with time complexity $\mathcal{O}(n2^{2n})$ and memory complexity $\Theta(2^{2n})$ using only a few ancilla qubits, but also demonstrates superior resource trade-offs. Our proposed block-encoding protocol is named Binary Tree Block-encoding (\texttt{BITBLE}). Under the benchmark, \textit{size metric}, defined by the product of the number of gates and the normalization factor, numerical experiments demonstrate the improvement of both resource trade-off and classical computing time efficiency of the \texttt{BITBLE} protocol. The algorithms are all open-source.
\end{abstract}




\section{Introduction}
\label{section Introduction}
\IEEEPARstart{Q}{uantum} signal processing (QSP)~\cite{PhysRevLett.118.010501}, quantum singular value transformation (QSVT)~\cite{10.1145/3313276.3316366,PRXQuantum.2.040203} are powerful frameworks for solving high-dimensional eigenvalue and singular value problems. Under this framework, eigenvalue or singular value problems can be solved without solving the eigenvalue decomposition and singular value decomposition, providing an alternative way to solve numerical problems. While the quantum circuit's simulation is difficult as the time and memory trade-off of the simulation improves exponentially with the qubit number $n$ as $\mathcal{O}(2^n)$. 

The input model employed in these methods is based on block-encoding. Leveraging the block-encoding framework, several seminal quantum algorithms---such as Hamiltonian simulation~\cite{PhysRevLett.118.010501,osti_1609315}, Grover's algorithm, the quantum Fourier transform, and the HHL algorithm~\cite{PhysRevLett.120.050502,PhysRevLett.103.150502}---can be interpreted as instances of the Quantum Singular Value Transformation (QSVT)~\cite{PRXQuantum.2.040203}. The embedding of a matrix \( A \) is typically realized as the leading principal block of a larger unitary matrix \( U \) acting on the Hilbert space, expressed as:
$$
U = \begin{bmatrix}
    A & * \\
    * & *
\end{bmatrix},
$$
where \( * \) denotes arbitrary matrix elements. However, if \( \Vert A\Vert_2 \geq 1 \), such an embedding is impossible. To address this limitation, a formal definition of a block-encoding for an \( n \)-qubit matrix \( A \) in an \( m \)-qubit system is provided as follows~\cite{10.1145/3313276.3316366}:
\begin{definition}[\cite{10.1145/3313276.3316366}]
Let $a,n,m\in\mathbb{N}$ with $m=a+n$. Then an $m$-qubit unitary $U$ is a $(\alpha,a,\varepsilon)$-block-encoding of an $n$-qubit operator $A$ if
$$
\Vert A - \alpha \left(\bra{0}^{\otimes a}\otimes I_n\right) U \left(\ket{0}^{\otimes a}\otimes I_n\right)\Vert\leq \varepsilon. 
$$
\end{definition}
The parameters $(\alpha,a,\varepsilon)$ are, respectively, the normalization factor~\cite{PhysRevA.110.042427} (also called subnormalization~\cite{Sunderhauf2024blockencoding}, subnormalization factor~\cite{9951292}, or constant factors~\cite{10012045}) for encoding matrices of arbitrary norm, the number of ancilla qubits used in the block-encoding and epsilon. The normalization factor of block-encoding is a crucial parameter that influences the circuit depth in quantum algorithms utilizing quantum signal processing (QSP)~\cite{PhysRevLett.118.010501} and quantum singular value transformation (QSVT)~\cite{9951292,PRXQuantum.2.040203}. For a classical data matrix of size $2^n \times 2^n$, a block-encoding protocol~\cite{10012045} provides a procedure for converting classical data into block-encoding using a quantum random access memory (QRAM) model~\cite{PRXQuantum.2.020311}. Building on this procedure, block-encoding protocols with near-optimal gate complexities have been extensively discussed~\cite{Zhang2024}. However, these block-encoding protocols require at least $\mathcal{O}(2^n/n)$ ancillary qubits~\cite{Yuan2023optimalcontrolled}, making it infeasible to simulate such large-scale block-encodings on a classical computer.

To simulate block-encodings in quantum signal processing and quantum singular value transformation on a classical computer, a fast approximate quantum circuit for block-encodings (\texttt{FABLE}) was proposed by Camps et.~\cite{9951292}. \texttt{FABLE} circuits characterize the sparsity of a class of matrices in the Walsh–Hadamard domain, and they can be modified to accommodate highly compressible circuits to block-encode a certain subset of sparse matrices~\cite{kuklinski2024sfablelsfablefastapproximate}. However, these block-encodings have a high normalization factor proportional to $2^n$, which incurs a high classical trade-off when encoding a high-dimensional matrix. Block-encoding for matrices of product operators~\cite{PhysRevA.110.042427} and structured matrices~\cite{10012045} has also been explored, but these methods are tailored to specific matrices rather than general matrices.

Our main contribution is to provide fast computational methods for block-encoding classical matrices with few ancilla qubits, low computing time, and low quantum gate counts. These methods are based on new numerical algorithms for decoupling multiplexor operations~\cite{PhysRevLett.93.130502}.

For single- and two-qubit gate decomposition of multiplexor operations with $2^n$ parameters, we propose permutative demultiplexor and recursive demultiplexor, achieving $\mathcal{O}(n2^n)$ time complexity. These generalize uniformly controlled rotation~\cite{PhysRevLett.93.130502} and recursive cancellation of uniformly controlled rotation~\cite{10.1145/1120725.1120847,1201583}.

For encoding $2^n \times 2^n$ matrices, we introduce Binary Tree Block-Encoding (\texttt{BITBLE}), based on the above decompositions. Its parameter computation has $\mathcal{O}(n2^{2n})$ time complexity and $\Theta(2^{2n})$ memory complexity, with normalization factors $\|A\|_F$ or $\mu_p(A)$. The recursive multiplexor operation technique in \texttt{BITBLE} is compatible with parallel computation. Numerical experiments demonstrate its advantages even in serial execution, establishing \texttt{BITBLE} as the fastest known algorithm for decoupling block-encodings of general matrices.

The content of this article is organized as follows: In Section~\ref{section Introduction}, the motivation and results of multiplexor operations and block-encoding protocols are discussed, and relevant notations are introduced. In Section~\ref{section multiplexor operations decomposition}, two decomposition methods for multiplexor operations—permutative demultiplexor and recursive demultiplexor—are proposed. In Section~\ref{section bitble}, \texttt{BITBLE} protocols and parameter-finding methods are introduced, where the time- and memory-complexity for computing single-qubit parameters is proven. In Section~\ref{section Examples}, numerical results for simulating this protocol using several examples are provided. Finally, Section~\ref{section Conclusion} presents the conclusion. MATLAB implementations of the \texttt{BITBLE} protocols, developed using QCLAB~\cite{qclab}, are publicly available at \href{https://github.com/zexianLIPolyU/BITBLE-SIABLE_matlab}{https://github.com/zexianLIPolyU/BITBLE-SIABLE\_matlab}.

Without loss of generality, we assume in the remainder of this paper that the matrix size is $N \times N$ with $N = 2^n$. The notation $\beta_{\cdot,k}$, $\beta_{k,\cdot}$, and $\beta_{\cdot}$ stands for vectors as $\beta_{\cdot,k}\equiv (\beta_{1,k},\beta_{2,k},\cdots)$, $\beta_{k,\cdot}\equiv (\beta_{k,1},\beta_{k,2},\cdots)^T$ and $\beta_{\cdot}\equiv (\beta_{1},\beta_{2},\cdots)$. The symbol $(\beta_{k,1},\beta_{k,2},\cdots)^T$ stands for the transport of a row vector $(\beta_{k,1},\beta_{k,2},\cdots)$, and $(\beta_{k,1},\beta_{k,2},\cdots)^*$ stands for the conjugate transport of that vector. The symbol $U^\dagger$ stands for the conjugate transport of the unitary matrix $U$. 


\section{Decouple Multiplexor Operation}
\label{section multiplexor operations decomposition}

A multiplexor operation (multiplexed rotations) \textit{controlled-}$R_{\bm{\alpha}}^{[\beta_{j}]_{j=0}^{2^n-1}}$ can be represented in a mathematical form of 
\begin{equation}
    \sum_{j=0}^{2^n-1} \ket{j}\bra{j} \otimes R_{\bm{\alpha}}^{\beta_j} = \begin{pmatrix}
        R_{\bm{\alpha}}^{\beta_0} & & & \\
        & R_{\bm{\alpha}}^{\beta_1} & & & \\
        & & \ddots & \\
        & & & R_{\bm{\alpha}}^{\beta_{2^n-1}}
    \end{pmatrix},
    \label{eq Multiplexed Rotation}
\end{equation}
and the rotation matrix $R_{\bm{\alpha}}^{\beta}\in\mathbb{C}^{2\times 2}$ is given by
$$
    R_{\bm{\alpha}}^{\beta} = e^{i{\bm{\alpha}\cdot\bm{\sigma}} \beta/2} = I \cos\frac{\beta}{2} + i \bm{\alpha}\cdot\bm{\sigma}\sin\frac{\beta}{2},
$$
where $\bm{\alpha}\cdot\bm{\sigma} = a_x X + a_y Y + a_z Z$ involves the Pauli matrices $X = \begin{bmatrix}
    0 & 1\\ 1 & 0
\end{bmatrix}$, $Y = \begin{bmatrix}
    0 & -i\\ i & 0
\end{bmatrix}$ and $Z = \begin{bmatrix}
    1 & 0\\ 0 & -1
\end{bmatrix}$, and $(a_x,a_y,a_z)\in\mathbb{R}^3$ is a real unit vector.

Multiplexor operations can be implemented by single- and two-qubit gates. In this article, we denote the single- and two-qubit decomposition of a multiplexor operation with the controlled qubit in the lowest position of a circuit~\cite{PhysRevLett.93.130502} as `\textit{uniformly controlled rotation}'. The single-qubit rotation parameters $[\tilde{\beta}_{j}]_{j=0}^{2^n-1}$ in the  uniformly controlled rotation are transformed from $[\beta_{j}]_{j=0}^{2^n-1}$ through a linear system. The matrix $M^n$ in this system, of size $2^n \times 2^n$, is generated by the product of the Walsh-Hadamard transformation $H^{\otimes n}$ and a Gray permutation matrix $P_G$, where $P_G$ transforms $n$-bit binary ordering into $n$-bit Gray code ordering~\cite{9951292,PhysRevLett.93.130502}. That is, $M^n \equiv H^{\otimes n} P_G$. The linear system is given by
    \begin{equation}
        M^n \left([\tilde{\beta}_{j}]_{j=0}^{2^n-1}\right)^T = \left([\beta_{j}]_{j=0}^{2^n-1}\right)^T,
        \label{eq Walsh-Hadamard}
    \end{equation}
    where $\left([\tilde{\beta}_{j}]_{j=0}^{2^n-1}\right)^T = (\tilde{\beta}_{0}, \ldots, \tilde{\beta}_{2^n-1})^T$, $\left([\beta_{j}]_{j=0}^{2^n-1}\right)^T = (\beta_{0}, \ldots, \beta_{2^n-1})^T$, and $n \in \mathbb{N}_+$. An example of uniformly controlled rotation in size $n=2$ is shown in Fig.~\ref{fig urc}.

\begin{figure}[htbp]
	\centering
    \[\resizebox{0.5\textwidth}{!}{
	$
	\begin{myqcircuit}
		& \ctrlsq{1} & \qw \\
		& \ctrlsq{1} & \qw \\
		& \gate{R_{\bm{\alpha}}^{\left[\beta_j\right]_{j=0}^{3}}} & \qw \\
	\end{myqcircuit} =  
	\begin{myqcircuit}
		& \qw & \qw & \qw & \ctrl{2} & \qw & \qw & \qw & \ctrl{2} & \qw \\
		& \qw & \ctrl{1} & \qw & \qw & \qw & \ctrl{1} & \qw & \qw & \qw \\
		& \gate{R_{\bm{\alpha}}^{\tilde{\beta}_0}} & \targ & \gate{R_{\bm{\alpha}}^{\tilde{\beta}_1}} & \targ & \gate{R_{\bm{\alpha}}^{\tilde{\beta}_2}} & \targ & \gate{R_{\bm{\alpha}}^{\tilde{\beta}_3}} & \targ & \qw \\
	\end{myqcircuit}
	$
	}\]
    \hfill
    \begin{subfigure}[htbp]{0.49\textwidth}
    \centering
     \includegraphics[width=1.0\textwidth]{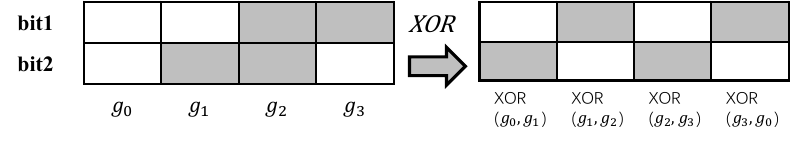}
    \end{subfigure}
	\caption{Uniformly controlled rotation decomposition~\cite{PhysRevLett.93.130502} for multiplexor operations of Eq.~\eqref{eq Multiplexed Rotation} with size $n=2$. The control nodes of the CNOT gates in this decomposition are determined by performing an \textsc{Exclusive Or} (XOR) operation on $n$-bit Gray codes. Specifically, these nodes correspond to the positions of the `$1$' bits (highlighted in gray) in the XOR result between two-bit Gray codes.}
    \label{fig urc}
\end{figure}

 \begin{lemma}[\cite{Amankwah2022}]
    The rotation parameters of the decomposition of a multiplexor operation controlled-$R_{\bm\alpha}^{[\beta_j]_{j=0}^{2^n-1}}$ can be calculated in classical computational time $\mathcal{O}(n2^n)$.
    \label{lemma time of uniformly controlled gates}
 \end{lemma}
 
 \begin{proof}
     The single-qubit gates' rotation parameters of uniformly controlled gates in Eq.~\eqref{eq Walsh-Hadamard} can be calculated in time $\mathcal{O}(n2^{n})$ by the scaled fast Walsh–Hadamard transform~\cite{Amankwah2022,1674569} and the Gray code permutation~\cite{9951292,Amankwah2022}.
 \end{proof}


 The uniformly controlled rotation assumes that the last node serves as the control node. However, if the control node of a multiplexor operation is not the last node in a quantum circuit, the mathematical representation, control nodes, and rotation parameters of the multiplexor operation will differ. In the following two subsections, we present two distinct decomposition methods for computing the rotation parameters of single-qubit gates in multiplexor operations where the control node is not the last node in a quantum circuit. 


\subsection{Permutative demultiplexor}
\label{subsecion multiplexor operations decomposition}



 The first decoupling method, \textit{permutative demultiplexor}, is derived from the permutation of indices in uniformly controlled rotations~\cite{PhysRevLett.93.130502}. 
 Consequently, the permutative demultiplexor can be implemented using an alternating sequence of CNOT gates and single-qubit rotations, similar to the  uniformly controlled rotation. The position of the control node in the $j$th CNOT gate of the permutative demultiplexor is determined by the `$1$' bit (highlighted in gray) in $\text{XOR}(g_j, g_{j+1})$, where $\{g_j\}$ represents the $j$th binary reflected Gray codes, and $\text{XOR}$ denotes the `Exclusive Or' logical operation applied to two binary codes. The single-qubit rotation parameters $\{\tilde{\beta}_j\}_{j=0}^{2^n-1}$ in a permutative demultiplexor can be transformed from $\{\beta_j\}_{j=0}^{2^n-1}$ by Eq.~\eqref{eq Walsh-Hadamard} as the same as the uniformly controlled rotation. An example of control nodes in permutative demultiplexor is illustrated in Fig.~\ref{fig multiplexor operation decomposition1}.

 \begin{figure*}
    \begin{subfigure}[htbp]{1.0\textwidth}
    \centering
     \includegraphics[width=0.8\textwidth]{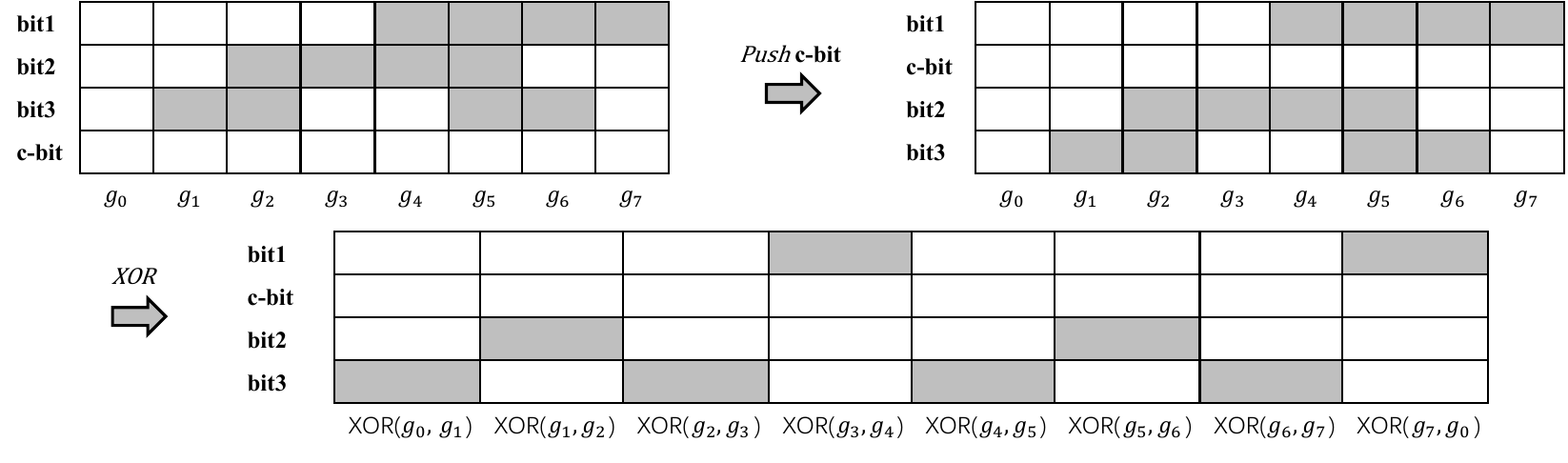}
    \caption{The binary reflected $3$-bit Gray code is used to define the positions of the control nodes. The counting order of `1' bit (highlighted in gray) in the transformed Gray code reveals the indices of the control nodes is $\{3, 2, 3, 1, 3, 2, 3, 1\}$.}
    \label{fig control index}
    \end{subfigure}
\\[\baselineskip]
    \begin{subfigure}[htbp]{1.0\textwidth}
    \[\begin{myqcircuit}
		& \ctrlsq{1} & \qw \\
		& \gate{R_{\bm{\alpha}}^{\left[\beta_j\right]_{j=0}^7}} & \qw \\
		& \ctrlsq{-1} & \qw \\
		& \ctrlsq{-1} & \qw \\
	\end{myqcircuit} \quad = \quad
	\begin{myqcircuit}
		& \qw & \qw & \qw & \qw & \qw & \qw & \qw & \ctrl{1} & \qw & \qw & \qw & \qw & \qw & \qw & \qw & \ctrl{1} & \qw \\
		& \gate{R_{\bm{\alpha}}^{\tilde{\beta}_{0}}} & \targ & \gate{R_{\bm{\alpha}}^{\tilde{\beta}_{1}}} & \targ & \gate{R_{\bm{\alpha}}^{\tilde{\beta}_{2}}} & \targ & \gate{R_{\bm{\alpha}}^{\tilde{\beta}_{3}}} & \targ & \gate{R_{\bm{\alpha}}^{\tilde{\beta}_{4}}} & \targ & \gate{R_{\bm{\alpha}}^{\tilde{\beta}_{5}}} & \targ & \gate{R_{\bm{\alpha}}^{\tilde{\beta}_{6}}} & \targ & \gate{R_{\bm{\alpha}}^{\tilde{\beta}_{6}}} & \targ & \qw \\
		& \qw & \qw & \qw & \ctrl{-1} & \qw & \qw & \qw & \qw & \qw & \qw & \qw & \ctrl{-1} & \qw & \qw & \qw & \qw & \qw \\
		& \qw & \ctrl{-2} & \qw & \qw & \qw & \ctrl{-2} & \qw & \qw & \qw & \ctrl{-2} & \qw & \qw & \qw & \ctrl{-2} & \qw & \qw & \qw \\
	\end{myqcircuit}\]
    \centering
    \caption{The quantum circuit realizes controlled-$R_{\bm{\alpha}}^{[\beta_{j}]_{j=0}^{7}}$ based on permutative demultiplexor, where the control node of CNOT gate behind the $j$th single-qubit rotation $R_{\bm{\alpha}}^{\tilde{\beta}_j}$ is determined by gray part of transformed Gray codes as described in Fig.~\ref{fig control index}, and the single-qubit rotation parameters $\{\tilde{\beta}_j\}_{j=0}^{7}$ are computed using Eq.~\eqref{eq Walsh-Hadamard}.}
    \label{subfig permutative multiplexor operations}
    \end{subfigure}
\caption{Permutative demultiplexor implement of controlled-$R_{\bm\alpha}^{[\beta_j]_{j=0}^{2^n-1}}$ with $k=1$ and $n=3$.}
\label{fig multiplexor operation decomposition1}
\end{figure*}


\subsection{Recursive demultiplexor}
\label{subsection Recursive demultiplexor}
    
  The second decoupling method, \textit{recursive demultiplexor}, leverages the recursive decomposition property of controlled multiplexor operations. As established in~\cite{10.1145/1120725.1120847}, a multiplexor operation acting on a target qubit (controlled by another qubit) can be decomposed into multiple simpler multiplexor operations interleaved with CNOT gates. 

   This leads to a key practical consequence: when a multiplexor operation appears in the middle of a quantum circuit (i.e., its control qubit is neither the first nor last qubit), the decomposition requires two recursive applications. Each recursion introduces additional CNOT gates and splits the original operation into progressively simpler multiplexor operations.

  Consider a multiplexor operation control-$R_{\bm\alpha}^{[\beta_j]_{j=0}^{2^n-1}}$ with $k$ control nodes in the upper part of the circuit and $n-k$ control nodes in the lower part of the circuit. The position of the control qubit in the recursive demultiplexor follows the specific rule.
  
  In the first recursion, the indices of the control qubits are determined by the binary reflected Gray code in the upper part of the circuit. An example of the first recursion of recursive demultiplexor (with $k=1,n=3$) is shown in Fig.~\ref{fig:The first recursion of recursive demultiplexor}. The rotation parameters for the multiplexor operation decomposition of the first recursion can be computed by
  \begin{equation}
      \begin{aligned}
        &\begin{bmatrix}
         \tilde\beta_{i,j}^{(1)}
      \end{bmatrix}^{i=1,\cdots,2^k}_{j=0,\cdots,2^{n-k}-1} \\
      &\quad = (M^k)^{-1}\text{reshape}\left(\begin{bmatrix}
        \beta_{j}
      \end{bmatrix}_{j=0}^{2^{n}-1},[2^k,2^{n-k}] \right),
      \end{aligned}
  \label{eq first recursion}
  \end{equation}
  where the matrix $M^k$ is a $2^k$ dimensional Walsh-Hadamard transform defined in Eq.~\eqref{eq Walsh-Hadamard}, and $reshape(\beta,[n_1,n_2])$ reshapes $\beta$ into a $n_1$-by-$n_2$ matrix;
  
  In the second recursion, the indices of the control qubits are determined by the binary reflected Gray code in the lower part of the circuit. The rotation parameters of the second recursion can be computed by
  \begin{equation}
  \begin{aligned}
    &\begin{bmatrix}
     \tilde\beta_{i,j}^{(2)}
      \end{bmatrix}^{i=1,\cdots,2^k}_{j=0,\cdots,2^{n-k}-1}\\
    &\quad   = \left[(M^{n-k})^{-1}\left[\begin{bmatrix}
         \tilde\beta_{i,j}^{(1)}
      \end{bmatrix}^{i=1,\cdots,2^k}_{j=0,\cdots,2^{n-k}-1}\right]^T \right]^T.
  \end{aligned}
  \label{eq second recursion}
  \end{equation}

  \begin{figure}
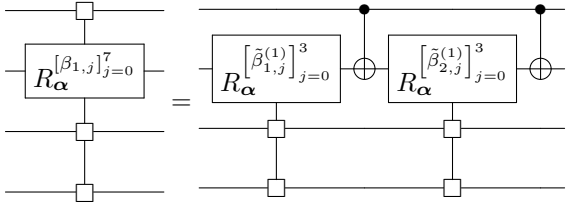

    \centering
    \[\begin{myqcircuit}
		& \ctrlsq{1} & \qw \\
		& \gate{R_{\bm{\alpha}}^{\left[\beta_{1,j}\right]_{j=0}^{7}}} & \qw \\
		& \ctrlsq{-1} & \qw \\
		& \ctrlsq{-1} & \qw \\
	\end{myqcircuit} = 
	\begin{myqcircuit*}{-0.3}{0.5}
		& \qw & \ctrl{1} & \qw & \ctrl{1} & \qw \\
		& \gate{R_{\bm{\alpha}}^{\left[\tilde{\beta}_{1,j}^{(1)}\right]_{j=0}^{3}}} & \targ & \gate{R_{\bm{\alpha}}^{\left[\tilde{\beta}_{2,j}^{(1)}\right]_{j=0}^{3}}} & \targ & \qw \\
		& \ctrlsq{-1} & \qw & \ctrlsq{-1} & \qw & \qw \\
		& \ctrlsq{-1} & \qw & \ctrlsq{-1} & \qw & \qw \\
	\end{myqcircuit*}\]
    \caption{The first recursion of recursive demultiplexor implements controlled--$R_{\alpha}^{[\beta_j]_{j=0}^7}$ operations ($k=1,n=3$) with the controlled qubit on the second qubit.}
    \label{fig:The first recursion of recursive demultiplexor}
\end{figure}

  The position of the control qubits of CNOTs in the recursive multiplexor operation can be determined by the `1' bit of transformed Gray codes. Specifically, the control nodes in the upper part are determined by the $k$-bit Gray codes composed of `\textbf{u-bit}', the control nodes in the lower part are determined by another $(n-k)$-bit Gray code composed of `\textbf{l-bit}', and the controlled node is denoted as `\textbf{c-bit}'. To determine the positions of the control nodes, the Exclusive Or logical operation (\text{XOR}) is applied to the binary codes in the upper and lower parts separately. Then, the codes in the lower part are broadcast based on each Gray code in the upper part. The position of the control node in the $j$th CNOT gate of the permutative demultiplexor is determined by the `$1$' bit after this broadcast operation. The structure contains $2^k$ idle single-qubit gates periodically placed between consecutive CNOTs, occurring every $2^{n-k}$ single-qubit rotations. An example of recursive multiplexor operation decomposition ($k=1,n=3$) is illustrated in Fig.~\ref{fig multiplexor operation decomposition2}.

 \begin{lemma}
    A multiplexor operation controlled-$R_{\bm\alpha}^{[\beta_j]_{j=0}^{2^n-1}}$ can be decoupled into single- and two-qubit gates using recursive demultiplexor in classical computational time $\mathcal{O}(n {2^n})$.
    \label{lemma time of permutative multiplexor operations}
 \end{lemma}
 
 \begin{proof}
    Since the operation `$\text{reshape}\left(\begin{bmatrix}
    \beta_{j}
    \end{bmatrix}_{j=0}^{2^{n}-1},[2^k,2^{n-k}] \right)$' takes time $\Theta(2^n)$, and the linear equation 
    $$
    \left(M^k\right) [\tilde\beta_{j}]_{j=0}^{2^k-1}= [\beta_{j}]_{j=0}^{2^k-1}
    $$
    can be solved by the fast Walsh–Hadamard transform and Gray permutation in time $\mathcal{O}(k2^k)$~\cite{Amankwah2022,1674569}. It takes a time of $k\times2^k\times2^{n-k} + (n-k)\times2^k\times2^{n-k} = \mathcal{O}(n2^n)$ to compute the parameters of the recursive demultiplexor in Eq.~\eqref{eq first recursion} and Eq.~\eqref{eq second recursion}. 
 \end{proof}

Both decoupling methods are compatible with implementations using only nearest-neighbor CNOT gates~\cite{PhysRevA.71.052330}.

\begin{figure*}[htbp]
    \centering
    \begin{subfigure}[htbp]{1.0\textwidth}
    \centering
    \includegraphics[width=0.8\textwidth]{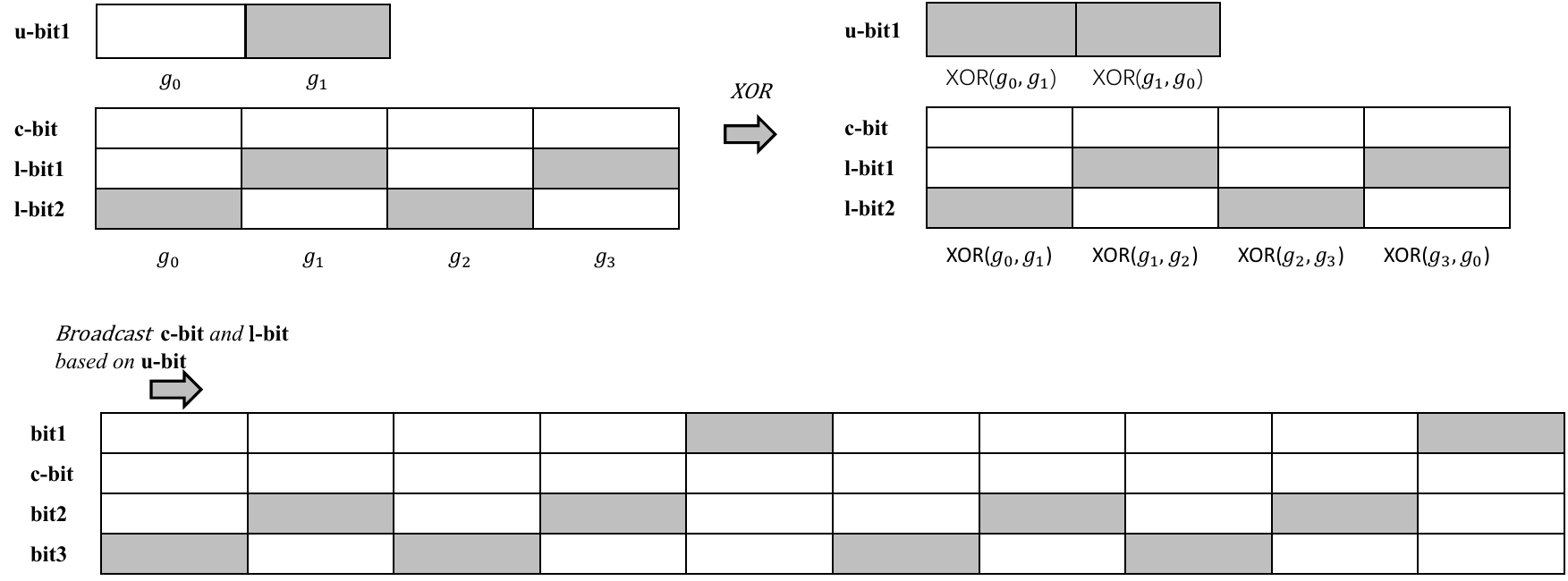}
    \caption{The binary reflected Gray code, denoted as `\textbf{u-bit}' in the upper part, and the binary reflected Gray code, denoted as `\textbf{l-bit}' in the lower part. The order of `1' bit (highlighted in gray) of the transformed Gray code indicates the indices of control nodes is $\{\text{3, 2, 3, 2, 1, 3, 2, 3, 2, 1}\}$.}
    \label{subfigure recursive multiplexor operation decomposition of lower part}
    \end{subfigure}
\\[\baselineskip]
    \begin{subfigure}[htbp]{1.0\textwidth}
    \centering
    \[
	\resizebox{\textwidth}{!}{%
	$
	\begin{myqcircuit}
		& \ctrlsq{1} & \qw \\
		& \gate{R_{\bm{\alpha}}^{\left[\beta_{1,j}\right]_{j=0}^7}} & \qw \\
		& \ctrlsq{-1} & \qw \\
		& \ctrlsq{-1} & \qw \\
	\end{myqcircuit}  = 
	\begin{myqcircuit}
		& \qw & \qw & \qw & \qw & \qw & \qw & \qw & \qw & \ctrl{1} & \qw & \qw & \qw & \qw & \qw & \qw & \qw & \qw & \ctrl{1} & \qw \\
		& \gate{R_{\bm{\alpha}}^{\tilde{\beta}_{1,0}^{(2)}}} & \targ & \gate{R_{\bm{\alpha}}^{\tilde{\beta}_{1,1}^{(2)}}} & \targ & \gate{R_{\bm{\alpha}}^{\tilde{\beta}_{1,2}^{(2)}}} & \targ & \gate{R_{\bm{\alpha}}^{\tilde{\beta}_{1,3}^{(2)}}} & \targ & \targ & \gate{R_{\bm{\alpha}}^{\tilde{\beta}_{2,0}^{(2)}}} & \targ & \gate{R_{\bm{\alpha}}^{\tilde{\beta}_{2,1}^{(2)}}} & \targ & \gate{R_{\bm{\alpha}}^{\tilde{\beta}_{2,2}^{(2)}}} & \targ & \gate{R_{\bm{\alpha}}^{\tilde{\beta}_{2,3}^{(2)}}} & \targ & \targ & \qw \\
		& \qw & \qw & \qw & \ctrl{-1} & \qw & \qw & \qw & \ctrl{-1} & \qw & \qw & \qw & \qw & \ctrl{-1} & \qw & \qw & \qw & \ctrl{-1} & \qw & \qw \\
		& \qw & \ctrl{-2} & \qw & \qw & \qw & \ctrl{-2} & \qw & \qw & \qw & \qw & \ctrl{-2} & \qw & \qw & \qw & \ctrl{-2} & \qw & \qw & \qw & \qw \\
	\end{myqcircuit}
	$}
	\]
    \caption{The second recursion of recursive demultiplexor implements controlled--$R_{\alpha}^{[\beta_j]_{j=0}^7}$ operations, where the control nodes of CNOT gate are determined by `1' bit (gray part) of the transformed Gray codes as described in Fig.~\ref{subfigure recursive multiplexor operation decomposition of lower part}. The parameters $\tilde\beta_{\cdot}^{(2)}$ are computed by Eq.~\eqref{eq first recursion} and Eq.~\eqref{eq second recursion}. In this structure, there exist two idle single-qubit gates between the $4$th and $5$th, $9$th and $10$th CNOT gates.}
    \label{fig idle single-qubit gate}
    \end{subfigure}
\caption{Recursive demultiplexor implement of controlled-$R_{\bm\alpha}^{[\beta_j]_{j=0}^{2^n-1}}$ with $k=1$ and $n=3$.}
\label{fig multiplexor operation decomposition2}
\end{figure*}


\section{Binary Tree Block-encoding}
\label{section bitble}

\subsection{Quantum state preparation by multiplexor operations}
\label{subsection quantum state preparation}

The objective of quantum state preparation is to generate a target quantum state \(\ket{\psi}\) from an initial product state \(\ket{0}^{\otimes n}\) using single- and two-qubit gates. A general quantum state can be expressed as
\begin{equation}
    \ket{\psi} = \sum_{k=0}^{N-1} e^{i\phi_k}|\psi_k|\ket{k},
    \label{eq ket psi}
\end{equation}
where $N = 2^n$, $\psi_k \in \mathbb{C}$, $\sum_{k=0}^{N-1} |\psi_k|^2 = 1$, and $\ket{k} \equiv \ket{k_n k_{n-1} \cdots k_1}$ represent the computational basis with bits $k_j$ for $j = 1, 2, \dots, n$. Numerous studies have investigated quantum state preparation~\cite{Low2024tradingtgatesdirty,Gui2024spacetimeefficient,10.1109/TCAD.2023.3297972,PhysRevLett.129.230504,SunXiaoming2023,10.5555/2011670.2011675,PhysRevResearch.3.043200,rosenthal2023querydepthupperbounds,Zhang2024}, demonstrating that for an $n$-qubit state $\ket{\psi}$ and a desired error precision $\epsilon$, a quantum state-preparation algorithm can produce an approximate state $\ket{\tilde{\psi}}$ satisfying $\Vert\ket{\psi} - \ket{\tilde{\psi}}\Vert \leq \epsilon$~\cite{Gui2024spacetimeefficient}.

The state preparation procedure using pre-computed amplitudes is well established in the literature~\cite{grover2002creatingsuperpositionscorrespondefficiently}. Additionally, the binary tree data structure for quantum states with real amplitudes was introduced by~\cite{kerenidis_et_al:LIPIcs.ITCS.2017.49}. In this work, we extend these methods to efficiently compute rotation parameters for general complex amplitudes with low time complexity. Quantum state preparation is described using a sequence of multiplexor operations, as illustrated in Fig.~\ref{quantum circuit Binary Tree}.

\begin{figure*}[htbp]
    \centering
    \[
    \begin{myqcircuit}
		& \gate{R_{Z}^{\theta_{-1}}} & \gate{R_{Y}^{\varphi_{0,0}}} & \ctrlsq{1} & \ctrlsq{1} & \gate{R_{Z}^{\theta_{0,0}}} & \ctrlsq{1} & \ctrlsq{1} & \qw \\
		& \qw & \qw & \gate{R_{Y}^{\left[\varphi_{1,j}\right]_{j=0}^1}} & \ctrlsq{1} & \qw & \gate{R_{Z}^{\left[\theta_{1,j}\right]_{j=0}^1}} & \ctrlsq{1} & \qw \\
		& \qw & \qw & \qw & \gate{R_{Y}^{\left[\varphi_{2,j}\right]_{j=0}^3}} & \qw & \qw & \gate{R_{Z}^{\left[\theta_{2,j}\right]_{j=0}^3}} & \qw \\
	\end{myqcircuit}\]
    \caption{Quantum circuit for state preparation using multiplexor operations in the case of $n=3$.}
\label{quantum circuit Binary Tree}
\end{figure*}

\begin{figure}
\begin{subfigure}[htbp]{0.49\textwidth} 
    \centering
    \includegraphics[width=1.0\linewidth]{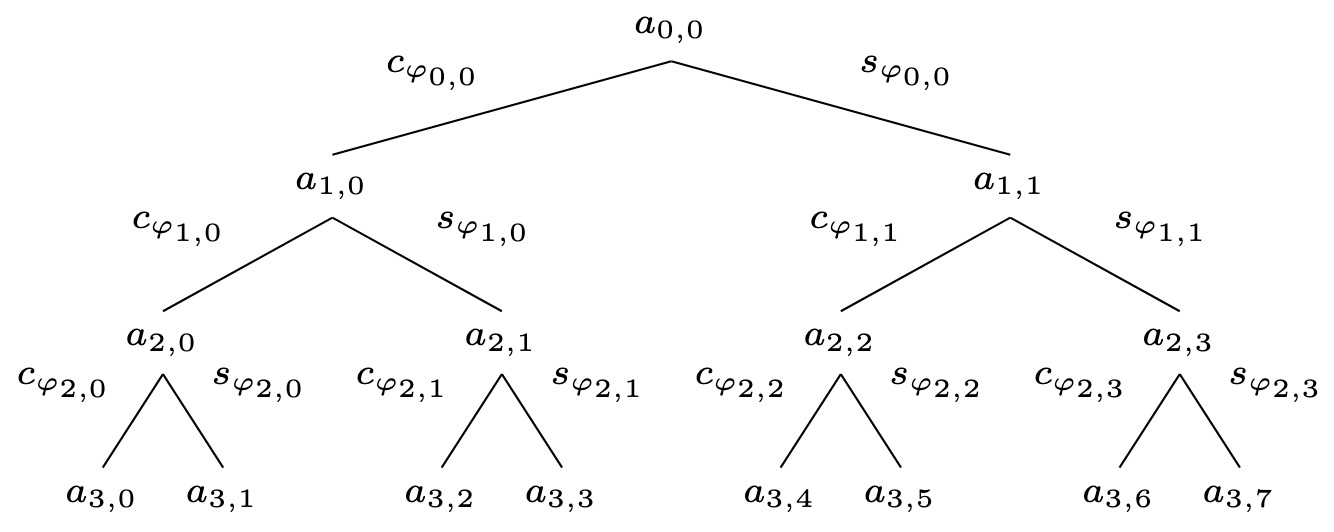}
    \caption{Rotation-$Y$ binary tree.}
    \label{fig:rotation-y binary tree}
\end{subfigure}
\hfill
\begin{subfigure}[htbp]{0.49\textwidth} 
    \centering
    \includegraphics[width=1.0\linewidth]{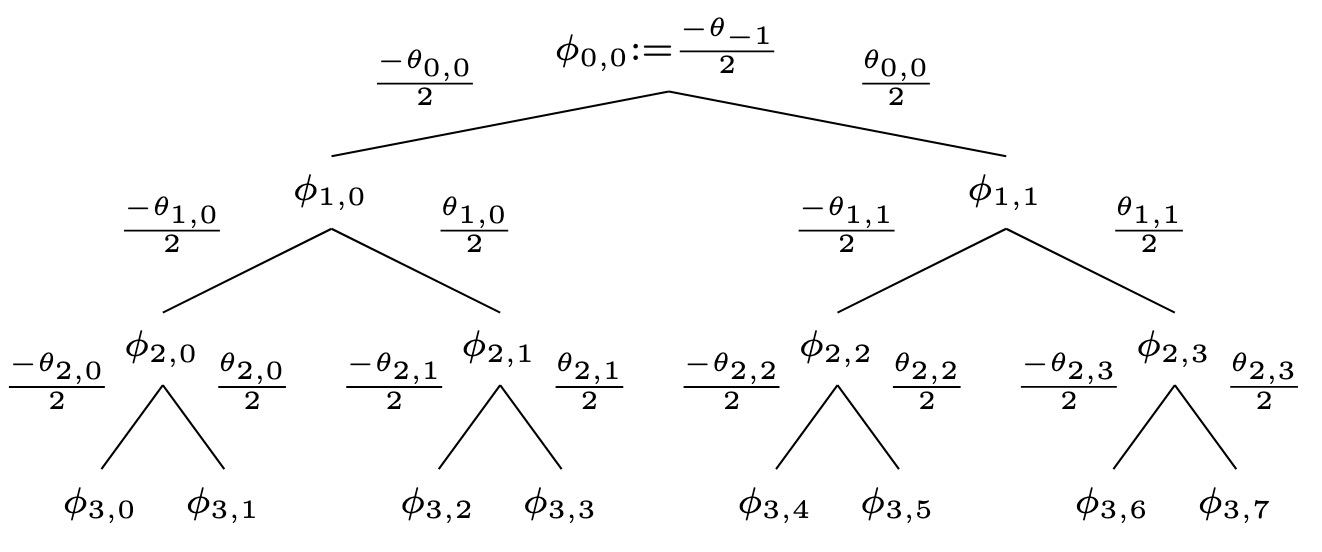}
    \label{rotation-Z binary tree}
     \caption{Rotation-$Z$ binary tree.}
     \label{fig:rotation-z binary tree}
     \end{subfigure}
     \caption{Rotation-$Y$ binary tree: Each leaf node element in rotation-$Y$ binary tree is the product of the elements on path of edges from the root node to the leaf node, where the value of the left-child edge be $c_{\varphi_{\cdot}} = \cos(\varphi_{\cdot}/2)$ and the value of the right-child edge be $s_{\varphi_{\cdot}} = \sin(\varphi_{\cdot}/2)$; Rotation-$Z$ binary tree:  Each leaf node element in rotation-$Z$ binary tree is the sum of elements on the path from the root node to the leaf node.}
\end{figure}

The rotation-$Y$ binary tree is used to generate the norms of the amplitudes $\left\{\vert\psi_k\vert\right\}_{k=0}^{2^n-1}$ (or $\left\{\psi_k\right\}_{k=0}^{2^n-1}$ for real amplitudes) under the computational basis $\{\ket{k}\}_{k=0}^{2^n-1}$. A quantum state with $2^n$ amplitudes can be generated by a rotation-$Y$ binary tree with $n$ layers. Let $\vert\psi_k\vert$ be the $k$-th leaf node $\{a_{n,k}\}_{k=0}^{2^n-1}$ in the $n$-th layer of the rotation-$Y$ binary tree. The value of a node in the $t$-th layer satisfies the product of its edge and its parent node's value, i.e.,
\begin{equation}
    a_{t,k} = \left\{
    \begin{aligned}
    &a_{t-1,\lfloor k/2\rfloor} \times  \cos(\varphi_{t-1,\lfloor k/2\rfloor}/2),\ &&\text{for even } k,\\
    & a_{t-1,\lfloor k/2\rfloor} \times  \sin(\varphi_{t-1,\lfloor k/2\rfloor}/2),\ &&\text{for odd } k,
    \end{aligned}
    \right.
    \label{eq multiple relation}
\end{equation}
for all $0 \leq t \leq n-1$ and $0 \leq k \leq 2^t - 1$. Equation~\eqref{eq multiple relation} leads to the relation
$$
\frac{\varphi_{t,k}}{2} = \text{angle}(a_{t+1,2k-1} + a_{t+1,2k}\cdot i),
$$
that is,
$$
\begin{aligned}
    e^{\frac{\varphi_{t,k}}{2} i} &= \cos\left(\frac{\varphi_{t,k}}{2}\right) + \sin\left(\frac{\varphi_{t,k}}{2}\right)\cdot i \\
    & = a_{t+1,2k-1} + a_{t+1,2k}\cdot i.
\end{aligned}
$$
The rotation-$Z$ binary tree is used to generate the phases $\left\{e^{i\phi_k}\right\}_{k=0}^{2^n-1}$ under the computational basis $\{\ket{k}\}_{k=0}^{2^n-1}$. A quantum state with $2^n$ amplitudes can be generated by a rotation-$Z$ binary tree with $n$ layers. Let $\phi_k$ be the $k$-th leaf node $\{\phi_{n,k}\}_{k=0}^{2^n-1}$ in the $n$-th layer of the rotation-$Z$ binary tree. The value of a node in the rotation-$Z$ binary tree is the sum of its edge and its parent node's value, i.e.,
\begin{equation}
    \phi_{t,k} = \phi_{t-1,\lfloor k/2\rfloor} + \frac{(-1)^{k+1}}{2} \theta_{t-1,\lfloor k/2\rfloor},
    \label{eq additive relation}
\end{equation}
for $1 \leq t \leq n$ and $0 \leq k \leq 2^t - 1$.

Since the sum of the squares of the amplitudes of a pure state is $1$, the amplitudes $\left\{\vert\psi_k\vert\right\}_{k=0}^{2^n-1}$ can be determined by $2^{n} - 1$ degrees of freedom in a rotation-$Y$ binary tree with $n$ layers. However, a rotation-$Z$ binary tree with $t$ layers provides only $2^{t} - 1$ degrees of freedom, which can be compensated by introducing a global phase $\theta_{-1}$ at the beginning, as shown in Fig.~\ref{quantum circuit Binary Tree}. The angles $\theta_{-1}$ and $\{\theta_{t,k}\}_{0 \leq t \leq n, 0 \leq k \leq 2^t - 1}$ above $n$ layers of the rotation-$Z$ binary tree can be solved by the linear system
\begin{equation}
    M_{R_Z}^t \begin{bmatrix}
        \theta_{-1}\\
        \theta_{0,0}\\
        \vdots\\
        \theta_{t-1,2^{t-1}-1}
    \end{bmatrix} = \begin{bmatrix}
        \phi_{t,0}\\
        \phi_{t,1}\\
        \vdots\\
        \phi_{t,2^{t}-1}
    \end{bmatrix},
    \label{eq linear system in rotation-Z bianry tree}
\end{equation}
for all $1 \leq t \leq n$, where the matrix elements of $M_{R_Z}^t$ are determined by Eq.~\eqref{eq additive relation}. The inverse matrix $\left(M_{R_Z}^t\right)^{-1}$ follows the recursive relationship
\begin{equation}
\begin{aligned}
        &\left(M_{R_Z}^1\right)^{-1} = \begin{bmatrix}
        -1 & -1\\
        -1 & 1\end{bmatrix},\\
        &\left(M_{R_Z}^t\right)^{-1} = \begin{bmatrix}
            \frac{1}{2} \left(M_{R_Z}^{t-1}\right)^{-1} & 0 \\
            0 & I_{2^{t-1}} \end{bmatrix} \left(\begin{bmatrix}
            -1 & -1\\
            -1 & 1\end{bmatrix} \otimes I_{2^{t-1}}\right),
\end{aligned}
\label{eq inverse Matrix in RZ}
\end{equation}
for $t = 2, \dots, n$.

\begin{figure}[htbp]
    \centering\includegraphics[width=0.5\linewidth]{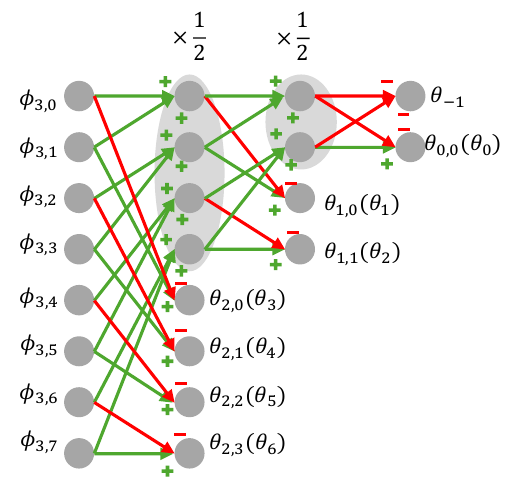}
    \caption{Computation process of R$_z$-angles with $n = 3$ layers.}
    \label{fig:RZangles}
\end{figure}

\begin{theorem}
    Given a quantum state $\ket{\psi} = \sum_{k=0}^{2^n-1} e^{i\phi_k}\vert \psi_k\vert\ket{k}\in\mathbb{C}^{n}$, the time complexity of computing the rotation-$Y$ and rotation-$Z$ parameters $\{\varphi_j\}_{j=0}^{2^n-1}$ and $\{\phi_j\}_{j=0}^{2^n-1}$ to prepare $\ket{\psi}$ is $\Theta(2^{n})$.
    \label{thm the time complextiy of computing rotation time}
\end{theorem}
\begin{proof}
    Since the multiplexor operation-$Y$ and multiplexor operation-$Z$ consist of $2^n$ parameters that can be constructed by $n$ layers, and the $l$-layer consists of $2^l$ edges which require $\Theta(2^l)$ computations. Therefore, it takes $\sum_{l=1}^n 2^l = \Theta(2^{n})$ time to compute all the rotation parameters.
\end{proof}


\subsection{Block-encoding protocols}

This protocol follows the prescription laid out in \cite{10.1145/3313276.3316366,10012045,kerenidis_et_al:LIPIcs.ITCS.2017.49,Chakrabarti2020quantumalgorithms}, which forms the unitary $U_A$ as the product of the controlled-state preparation unitary $U_R$, the state-preparation unitary $U_L$ and swap gates. In this prescription,
\begin{equation}
    U_A = (U_L^\dagger\otimes I_{n_2})(\text{SWAP}_{n_1,n_2})U_R ,
    \label{eq block encoding state preparation}
\end{equation}
where, controlled by an $n$-qubit register in the state $\ket{j}_{n_2}$, $U_R$ prepares the $n$-qubit state $\ket{A_j}_{n_1}$, and $U_L$ prepares the state $\ket{A_{F}}_{n_2}$ with $n_q$ ancilla qubits. 
\begin{equation}
\begin{aligned}
    &U_R\ket{0}_{n_1}\ket{0}_{n_q}\ket{j}_{n_2} = \ket{A_j}_{n_1} \ket{0}_{n_q}\ket{j}_{n_2},\\
    &U_L\ket{0}_{n_1}\ket{0}_{n_q} = \ket{A_F}_{n_1}\ket{0}_{n_q}.\\
\end{aligned}
\end{equation}
There are different normalization factors in the block-encoding protocol in different state-preparation based on equation~\eqref{eq block encoding state preparation}. First, if we define the states as 
\begin{equation}
\begin{aligned}
    &\ket{A_j}_{n_1} = \sum_{k=0}^{2^n-1}\frac{A_{k,j}}{\Vert A_{\cdot j}\Vert}\ket{k}_{n_1}, \\
    &\ket{A_F}_{n_1} = \sum_{j=0}^{2^n-1}\frac{\Vert A_{\cdot, j}\Vert}{\Vert A\Vert_F}\ket{j}_{n_1},
\end{aligned}
\end{equation}
where $\Vert A_{\cdot, j}\Vert$ is the Euclidean norm of $j$th column of $A$, $\Vert A\Vert_F$ is the Frobenius norm of $A$, then $U_A$ in equation~\eqref{eq block encoding state preparation} is an $(\Vert A\Vert_F,n+n_q)$-block-encoding of $A$. That is,
$$
\bra{0}_{n_1}\bra{0}_{n_q}\bra{k}_{n_2}U_A \ket{0}_{n_1}\ket{0}_{n_q}\ket{j}_{n_2} = \frac{A_{k,j}}{\Vert A\Vert_F}.
$$

Second, if we define $U_R$ and $U_L$ as~\cite{10012045,kerenidis2017} 
\begin{equation}
    \begin{aligned}
    &U_R\ket{0}_{n_1}\ket{0}_{n_q+2}\ket{j}_{n_2} = \ket{A_j^{(p)}}_{n_1 +2}\ket{0}_{n_q}\ket{j}_{n_2},\\ &U_L\ket{0}_{n_1}\ket{0}_{n_q+2}\ket{k}_{n_2} = \ket{\tilde{A}_k^{(p)}}_{n_1 +2}\ket{0}_{n_q}\ket{k}_{n_2},
    \end{aligned}
    \label{eq BITBLE extended}
\end{equation}
for $p\in[0,1]$, where 
$$
\begin{aligned}
&\ket{A_j^{(p)}}_{n_1+2} =\\
&\qquad \sum_{k\in[N]} \frac{e^{i\theta_{k,j}}\vert A_{k,j}\vert^p}{\sqrt{\Vert A_{\cdot,j}\Vert_{2p}^{2p}}}\ket{k}_{n_1} \left[\cos\chi_j\ket{0} + \sin\chi_j\ket{1} \right] \ket{0}, \\
&\ket{\tilde{A}_k^{(p)}}_{n_1+1} = \\
&\qquad \sum_{j\in[N]}\frac{\vert A_{k,j}\vert^{1-p}}{\sqrt{\Vert A_{k,\cdot}\Vert_{2(1-p)}^{2(1-p)}}}\ket{j}_{n_1}\ket{0}\left[\cos\chi_k\ket{0} + \sin\chi_k\ket{1} \right],
\end{aligned}
$$
 $\cos\chi_{j} = \sqrt{\frac{\Vert A_{\cdot,j}\Vert_{2p}^{2p}}{S_{2p}(A^T)}}$, $\cos\chi_k = \sqrt{\frac{\Vert A_{k,\cdot}\Vert_{2(1-p)}^{2(1-p)}}{S_{2(1-p)}(A)}}$, and $S_q(A)=\max_{k}\Vert A_{k,\cdot}\Vert_q^q$ is the $q$th power of the maximum $q$-norm of any row of $A$, then $U_L^\dagger(\text{SWAP}_{n_1,n_2})U_R$ is a $(\mu_p(A^T),n+n_q+2)$-block-encoding of $A$. The normalization factor is $\mu_p(A^T)= \sqrt{S_{2p}(A^T)S_{2(1-p)}(A)}$. In this article, we will give  fast classical computation methods of circuit synthesis of $(\Vert A\Vert_F,n)$-block-encoding and $(\mu_p(A^T),n+2)$-block-encoding.


\subsection{Parameter finding}
\label{subsection Parameter finding}
The block-encoding protocol, named Binary Tree Block-encoding (\texttt{BITBLE}), is based on equation~\eqref{eq block encoding state preparation} using quantum state preparation through multiplexor operations. One of the protocols is a $(\Vert A\Vert_F,n)$-block-encoding of $A$ as shown in Fig.~\ref{fig bitble quantum circuit}, where controlled-$O_A$ corresponds to the controlled state preparation $U_R$ in equation~\eqref{eq block encoding state preparation}, and $O_G^\dagger$ corresponds to the state preparation $U_L^\dagger$.
\begin{figure*}[htbp]
    \begin{subfigure}[htbp]{1.0\textwidth}
    \centering
    \[\begin{myqcircuit*}{1}{1.5}
		&&& U_{R} && U_{L}^{\dagger} \\
		& \lstick{\ket{0}^{\otimes n}} & {/} \qw & \gate{O_{A}} & \qswap & \gate{O_{G}^{\dagger}} & \meter \\
		& \lstick{n\mbox{-qubit state } \ket{\psi}} & {/} \qw & \ctrlsq{-1} & \qswap \qwx[-1] & \qw & \rstick{\frac{A\ket{\psi}}{\left\|\ket{A}\right\|}} 
		\gategroup{2}{4}{3}{4}{.7em}{--}
		\gategroup{2}{6}{2}{6}{.7em}{--}
	\end{myqcircuit*}\]
    \caption{Circuit frameworks of \texttt{BITBLE}.}
    \end{subfigure}
    \\[\baselineskip]
    \begin{subfigure}[htbp]{1.0\textwidth}
    \centering
    \[\resizebox{0.7\textwidth}{!}{
	$
	\begin{myqcircuit*}{2.25}{0.75}
		&& \qw & \multigate{2}{O_{A}} & \qw \\
		& \lstick{n-2\mbox{ qubits}} & {/} \qw & \ghost{O_{A}} & \qw \\
		&& \qw & \ghost{O_{A}} & \qw \\
		& \lstick{n\mbox{-qubit state} \ket{\psi}} & {/} \qw & \ctrlsq{-1} & \qw \\
	\end{myqcircuit*} 
	 = 
	\begin{myqcircuit}
		& \qw & \gate{R_{Z}^{\left[\hat{\theta}_{-1,j}\right]_{j=0}^{2^n-1}}} & \gate{R_{Y}^{\left[\hat{\varphi}_{0,j}\right]_{j=0}^{2^n-1}}} & \qw & {\cdots} && \ctrlsq{1} & \gate{R_{Z}^{\left[\hat{\theta}_{0,j}\right]_{j=0}^{2^n-1}}} & \qw & {\cdots} && \ctrlsq{1} & \qw \\
		& {/} \qw & \qw & \qw & \qw & {\cdots} && \ctrlsq{1} & \qw & \qw & {\cdots} && \ctrlsq{1} & \qw \\
		& \qw & \qw & \qw & \qw & {\cdots} && \gate{R_{Y}^{\left[\hat{\varphi}_{n-1,j}\right]_{j=0}^{2^{2n-1}-1}}} & \qw & \qw & {\cdots} && \gate{R_{Z}^{\left[\hat{\theta}_{n-1,j}\right]_{j=0}^{2^{2n-1}-1}}} & \qw \\
		& {/} \qw & \ctrlsq{-3} & \ctrlsq{-3} & \qw & {\cdots} && \ctrlsq{-1} & \ctrlsq{-3} & \qw & {\cdots} && \ctrlsq{-1} & \qw \\
	\end{myqcircuit}
	$}
	\]
    \caption{State preparation unitary $U_R$.}
    \end{subfigure}
    \\[\baselineskip]
    \begin{subfigure}[htbp]{1.0\textwidth}
    \centering
    \[\resizebox{0.35\textwidth}{!}{$
    \begin{myqcircuit*}{2}{0.75}
		&& \qw & \multigate{2}{O_{G}^{\dagger}} & \qw \\
		&\lstick{n-2\mbox{ qubits}} & {/} \qw & \ghost{O_{G}^{\dagger}} & \qw \\
		&& \qw & \ghost{O_{G}^{\dagger}} & \qw \\
	\end{myqcircuit*} 
     = 
	\begin{myqcircuit}
		& \qw & \ctrlsq{1} & \qw & {\cdots} && \gate{R_{Y}^{-\hat{\varphi}'_{0,0}}} & \qw \\
		& {/} \qw & \ctrlsq{1} & \qw & {\cdots} && \qw & \qw \\
		& \qw & \gate{R_{Y}^{-\left[\hat{\varphi}'_{n-1,j}\right]_{j=0}^{2^{2n-1}-1}}} & \qw & {\cdots} && \qw & \qw \\
	\end{myqcircuit}
	$}\]
    \caption{Inverse state preparation unitary $U_L^\dagger$}
    \end{subfigure}
\caption{Quantum circuits of Binary Tree Block Encoding (\texttt{BITBLE}) with normalization factor $\Vert A\Vert_F$.}
\label{fig bitble quantum circuit}
\end{figure*} 
The complete parameters' computation follows these three processes:\\

 \textbf{Process 1. From classcial data to multiplexor operations' parameters.}
 
 The complete algorithms of computing the \texttt{BITBLE} protocol's rotation parameters $\{\varphi_{\cdot}\}$, $\{\varphi_{\cdot}'\}$, and $\{\theta_{\cdot}\}$ is provided in Algorithm~\ref{alg parameter bitble} in Appendix~\ref{appendix bitble}. The rotation parameters $\{\hat\varphi_\cdot\}$ and $\{\hat\theta_\cdot\}$ are the permutation of $\{\varphi_\cdot\}$ and $\{\theta_\cdot\}$, that is, $[\hat\theta_{-1,j}]_{j=0}^{2^n-1} \equiv [\theta_{-1,j}]_{j=0}^{2^n-1}$ and
 \begin{equation}
    \begin{aligned}
    \left[\hat\varphi_{k,j'}\right]_{j'=0}^{2^{n+k}-1} & \equiv \left[[\varphi_{2^{k}-1,j}]_{j=0}^{2^n-1},\ldots,[\varphi_{2^{k+1}-2,j}]_{j=0}^{2^n-1}\right],\\
    \left[\hat\theta_{k,j'}\right]_{j'=0}^{2^{n+k}-1} & \equiv \left[[\theta_{2^{k}-1,j}]_{j=0}^{2^n-1},\ldots, [\theta_{2^{k+1}-2,j}]_{j=0}^{2^n-1}\right],\\
    \end{aligned}
    \label{eq otation parameters}
\end{equation}
for all $0\leq k\leq n-1$.

 \textbf{Process 2. From multiplexor operations' parameters to single-qubit rotations' parameters.}

 There are two decoupling methods of multiplexor operations in \texttt{BITBLE} protocol: permutative demultiplexor and recursive demultiplexor. 

1) Decoupling by permutative demultiplexor in Section~\ref{subsecion multiplexor operations decomposition}. The single-qubit rotation parameters can be computed as
    \begin{equation}
        \begin{aligned}
            & \left([\tilde{\beta}_{-1,j}]_{j=0}^{2^n-1}\right)^T = \left(M^{n}\right)^{-1}\left([ \hat{\beta}_{-1,j} ]_{j=0}^{2^n-1}\right)^T, \\
            & \left([\tilde{\beta}_{k,j}]_{j=0}^{2^{n+k}-1}\right)^T = \left(M^{n+k}\right)^{-1}\left([ \hat{\beta}_{k,j} ]_{j=0}^{2^{n+k}-1}\right)^T,
        \end{aligned}
        \label{eq decomposition for multiplexor operation type1}
    \end{equation}
    for all $0\leq k\leq n-1$, where $\hat\beta_{-1,j}$ and $\hat\beta_{k,j}$, for $\hat\beta_{\cdot} = \hat\varphi_{\cdot}$ in the multiplexor operations-$Y$ (or $\hat\beta_{\cdot} = \hat\theta_{\cdot}$ in the multiplexor operations-$Z$). The solving process can be accelerated by the fast Walsh-Hadamard transformation~\cite{9951292,1674569}. 
    
2) Decoupling by recursive demultiplexor described in Section~\ref{subsection Recursive demultiplexor}. The single-qubit rotation parameters $\begin{bmatrix}
        [\tilde\beta_{k,j}^{(2)}]_{j=0}^{2^n-1} 
    \end{bmatrix}_{k=0}^{2^n-2}\in\mathbb{C}^{(2^n-1)\times 2^n}$ (abbreviated as $[\tilde{\beta}_{k,j}^{(2)}]$) can be computed as  
\begin{equation}
    [\tilde{\beta}_{k,j}^{(2)}] = \left(
        \left(M^n\right)^{-1}
        \begin{bmatrix}
            [\beta_{0,j}]_{j=0}^{2^n-1}  \\
            \left(M^1\right)^{-1} 
            \begin{bmatrix}
                [\beta_{1,j}]_{j=0}^{2^n-1} \\
                [\beta_{2,j}]_{j=0}^{2^n-1}
            \end{bmatrix}  \\
            \vdots \\
            \left(M^{n-1}\right)^{-1} 
            \begin{bmatrix}
                [\beta_{2^{n-1}-1,j}]_{j=0}^{2^n-1} \\
                \vdots \\
                [\beta_{2^{n}-2,j}]_{j=0}^{2^n-1}
            \end{bmatrix}
        \end{bmatrix}^T
    \right)^T,
    \label{eq decomposition for multiplexor operation type2}
\end{equation}
where the rotation parameters in the first bracket $(\cdot)$ correspond to the first recursion, and those in the second bracket $[\cdot]$ correspond to the second recursion. The process is parallelizable and can theoretically achieve up to a quadratic speedup over permutative demultiplexor when sufficient processing resources are available.


Actually, it decouples a large multiplexor operation into multiple multiplexor operations, this method requires $(2^{n-1} - 2)$ additional CNOT gates for encoding a $2^n \times 2^n$ matrix compared to the previous decoupling method using permutative demultiplexor.

 \textbf{Process 3. Circuit compression.}
 
 The circuit compression is inherited from~\cite{9951292}. Assume that the classical data obtained after the aforementioned two processes takes the form $\{\tilde{\beta}_{l,k}\}$. For elements satisfying $\vert\tilde{\beta}_{l,k}\vert \leq \delta$, where $\delta$ is a predefined threshold, these can be deemed negligible, and the single-qubit gate can be removed. In addition, the consecutive CNOT gates can be canceled out using a parity check.

 
\subsection{Time complexity and space complexity of parameter finding}

\begin{theorem}[Time of single-qubit gates' parameters computation in \texttt{BITBLE} protocol]
    The time complexity to calculate the single-qubit gates' parameters with \texttt{BITBLE} protocol of a $2^n\times 2^n$ matrix is $\mathcal{O}(n2^{2n})$. 
    \label{lemma BITBLE permutative multiplexor operation}
\end{theorem}
\begin{proof}
    We calculate the time complexity step by step:
    \begin{itemize}
        \item Multiplexor operations' parameters: By Theorem~\ref{thm the time complextiy of computing rotation time}, $2^n\times2^n$ multiplexor operations' parameters can be computed in time $\Theta(2^{2n})$. Therefore, the multiplexor operations' parameters in the \texttt{BITBLE} protocol can be generated in time $\Theta(2^{2n})$.
        \item Decoupling multiplexor operations' parameters: On one hand, it takes a time of $\sum_{k=0}^{n-1}(k + n)2^{n+k} = \mathcal{O}(n2^{2n})$ to decouple the recursive demultiplexor by Lemma~\ref{lemma time of uniformly controlled gates} and equation~\eqref{eq decomposition for multiplexor operation type1}; On the other hand, it takes  $\sum_{k=0}^{n}k\times2^{k}\times2^{n} =  \mathcal{O}(n2^{2n})$ time to decouple the permutative demultiplexor by Lemma~\ref{lemma time of permutative multiplexor operations} and equation~\eqref{eq decomposition for multiplexor operation type2}.  
        \item Circuit compression: It takes a time of $\Theta(2^{2n})$ to compress $\Theta(2^{2n})$ parameters and parity-check at most $\Theta(2^{2n})$ on secutive CNOT gates.
    \end{itemize}

    Above all, it takes $\mathcal{O}(n2^{2n})$ time to calculate all single-qubit parameters in \texttt{BITBLE} protocol.
\end{proof}

 Since the memory of the parameters $\beta_\cdot$ can be overwritten by  $\tilde{\beta}_{\cdot}$ in the process of computing, the space complexity of the multiplexor operation decomposition by equation~\eqref{eq decomposition for multiplexor operation type1} and equation~\eqref{eq decomposition for multiplexor operation type2} is both $\Theta(2^{2n})$, that is an optimal space complexity of parameter finding for encoding a $2^n\times 2^n$ size matrix.


\section{Numerical experiments}
\label{section Examples}

 In this section, some numerical experiments will be provided to demonstrate the efficiency and simulation performance of our block-encoding protocols. We put another block-encoding protocol -- Fast Approximate Quantum Circuits (\texttt{FABLE}~\cite{9951292}) as a benchmark, the time complexity, memory complexity and normalization factor of these protocols are shown in Table~\ref{tab: bitble_siable_fable}. \texttt{BITBLE} protocols have a lower normalization factor compared to \texttt{FABLE} circuits, since $\max\{\Vert A\Vert_F,\mu_p(A)\}\leq 2^n\max_{i,j}\vert A_{i,j}\vert$ for $p\in[0,1]$. The time complexity, memory complexity of computing the rotation parameters, and the normalization factor for these two block-encoding protocols are shown in Table~\ref{tab: bitble_siable_fable}. 
 \begin{table}[htbp]
 \centering
    \begin{tabular}{l | cccc}
    \hline \cline{1-4}
        Protocols & Time & Memory & Normalization factor  \\ \hline \cline{1-4}
        \texttt{BITBLE} & $\mathcal{O}(n2^{2n})$  & $\Theta(2^{2n})$ & $\Vert A \Vert_F$ or $\mu_p(A)$  \\ \hline
        \texttt{FABLE}~\cite{9951292} & $\mathcal{O}(n2^{2n})$ & $\Theta(2^{2n})$  & $2^n\max_{i,j}\vert A_{i,j}\vert$ \\  \hline \cline{1-4}
    \end{tabular}
 \caption{The time and memory complexity of computing parameters, as well as the normalization factor of two fast block-encoding protocols.}
 \label{tab: bitble_siable_fable}
 \end{table}

The experiments explore different block-encoding circuits with different decoupling methods and normalization factors. The normalization factor, decoupling methods and ancillary qubits of \texttt{BITBLE}\textsuperscript{1}, \texttt{BITBLE}\textsuperscript{2} and \texttt{BITBLE}\textsuperscript{3} protocols are shown in Table~\ref{tab: Maximum rotation parameters}. In the following experiments, \texttt{BITBLE}\textsuperscript{3} protocols set $p = 0.5.$

 \begin{table}[htbp]
    \centering
    \begin{tabular}{l|clc}
        \hline \cline{1-4}
        Protocols & \makecell{Normalization\\ Factor} & \makecell{Decoupling Methods} & Ancilla  \\ \hline \cline{1-4}
        \texttt{BITBLE}$^{1}$ & $\Vert A\Vert_F$ & \makecell{Recursive  demultiplexor} & $n$  \\ \hline
        \texttt{BITBLE}$^{2}$ & $\Vert A\Vert_F$ &  \makecell{Permutative demultiplexor} & $n$  \\ \hline
        \texttt{BITBLE}$^{3}$ & $\mu_p(A)$ & \makecell{Recursive  demultiplexor} & $n+2$ \\ \hline
    \end{tabular}
    \caption{Normalization factor, decoupling methods and ancilla qubits in block-encoding protocols of size $A\in\mathbb{C}^{2^n\times 2^n}$.}
    \label{tab: Maximum rotation parameters}
\end{table}


\subsection{Classical Circuit Synthesis Time for Block-Encoding of Random Matrices}

 Numerical experiments for encoding random matrices are presented in Table~\ref{tab: time of computing single-qubit rotations} and Fig.~\ref{fig:time of random matrix}. The experiments compare the decoupling time of random matrices using three methods:  
\begin{itemize}  
    \item \texttt{BITBLE} (our proposed approach),  
    \item \texttt{FABLE}, and  
    \item Qiskit's unitary synthesis~\cite{javadiabhari2024quantumcomputingqiskit} (which encodes $2^{n} \times 2^{n}$ random matrices into the top-left block of a random $2^{n+1} \times 2^{n+1}$ unitary).  
\end{itemize}  
Although decoupling multiplexor operations can be parallelized using permutative demultiplexor, our experiments were conducted in serial computing. Even under serial computation, permutative demultiplexor demonstrate significant advantages. Among these, \texttt{BITBLE}\textsuperscript{1} achieves the fastest decoupling of multiplexor operations and computation of single-qubit rotation parameters.

 \begin{table}[htbp]
    \centering
    \begin{tabular}{l | ccccc}
    \hline \cline{1-6}
        $n$ & \texttt{BITBLE}$^{1}$  & \texttt{BITBLE}$^{2}$ & \texttt{BITBLE}$^{3}$ & \texttt{FABLE}~\cite{9951292} & Qiskit \\ \hline \cline{1-6}
        5 & 0.017 & 0.010 & 0.014 & \textbf{0.008} & 6.97 \\ \hline
        6 & 0.031 & 0.035 & 0.054 & \textbf{0.028} & 46.3 \\ \hline
        7 & \textbf{0.104} & 0.113 & 0.211 & 0.112 & 221.2 \\ \hline
        8 & 0.642 & \textbf{0.464} & 0.872 & 0.484 & 669.7 \\  \hline 
        9 & \textbf{1.683} & 1.980 & 3.637 & 1.984 & 767.4 \\  \hline 
        10 & \textbf{7.137} & 8.368 & 15.21 & 8.604 & 7262 \\ \hline
        11 & \textbf{29.35} & 35.20 & 65.04 & 36.49 & - \\  \hline 
        12 & \textbf{120.4} & 147.1 & 269.5 & 148.7 & - \\ \hline
        13 & \textbf{498.6} & 623.6 & 1125 & 625.5 & -\\ \hline
        14 & \textbf{2103} & 3119 & 5342 & 2831 & - \\ \hline
        \cline{1-6}
    \end{tabular}
     \caption{Time of classical circuit synthesis for random matrices of size $2^n \times 2^n$ (corresponding to $n$ qubits). \texttt{BITBLE}\textsuperscript{1} demonstrates significantly faster parameter computation compared to both \texttt{FABLE} and Qiskit's unitary synthesis. Notably, Qiskit's unitary synthesis encounters memory limitations when processing random matrices larger than $2^{10} \times 2^{10}$.}
     \label{tab: time of computing single-qubit rotations}
    \end{table}
 
  \begin{figure}[htbp]
    \centering
    \includegraphics[width=1.0\linewidth]{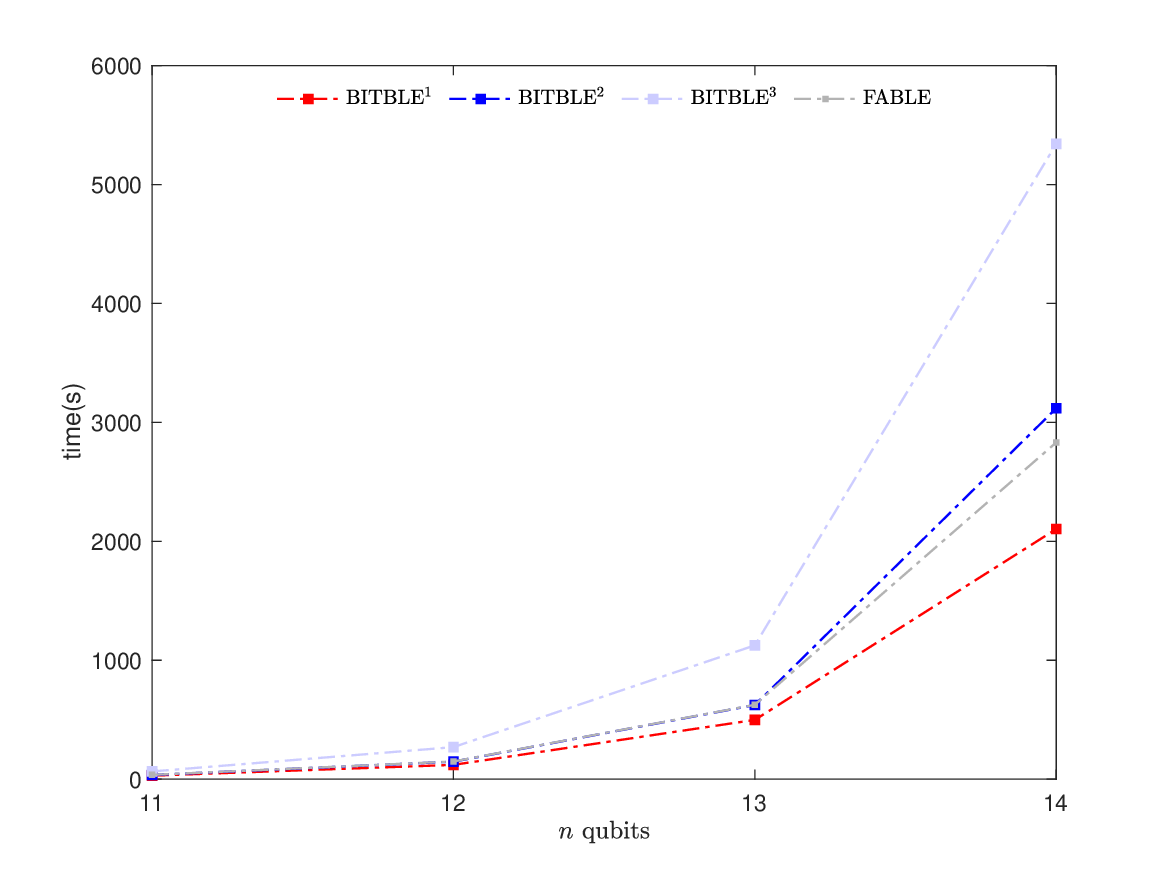}
    \caption{Time of classical circuit synthesis for $n$ qubits' random matrices, less time are preferable.}
    \label{fig:time of random matrix}
 \end{figure}


 \subsection{Size metric of block-encoding}
 
To further evaluate the performance of different block-encoding protocols, we introduce the \textit{size metric} for single-qubit rotation and two-qubit CNOT gates, defined as
\begin{equation}
    \textit{size\ metric} = \textit{number\ of\ gates} \times \textit{normalization\ factor}.
\end{equation}
 A special case of the above figure of metric, in terms of $T$-gate counts,  has been proposed by~\cite{Sunderhauf2024blockencoding}. Note that this metric is proportional to the normalization factor. In the remainder of this section, we present the CNOT and single-qubit rotation \textit{size metric} for \texttt{BITBLE} protocols applied to two model problems: image channels and discretized Laplacian operators in $1D$ and $2D$. Both simulations highlight the advantages of \texttt{BITBLE}.

\textbf{1) Encoding single-channel image.}

   The ``Peppers.png'' image and the FASHION MNIST dataset~\cite{zalandoresearch_fashionmnist} are used to evaluate the performance of different simulation-friendly block-encoding protocols. Additionally, we consider the unitary synthesis encoding $A$ into unitary
    $$
    \begin{bmatrix}
        A & \sqrt{I - AA^*} \\
        -\sqrt{I - A^*A} & A^*
    \end{bmatrix}
    $$
    implemented in Qiskit~\cite{javadiabhari2024quantumcomputingqiskit}.
    
    We compare the CNOT \textit{size metric} and parameter computation time across different methods, including the \texttt{FABLE} protocol (implemented via the QPIXL framework~\cite{Amankwah2022}) and Qiskit's unitary synthesis. We examine two distinct cases:
    \begin{enumerate}
        \item Encoding a ``Peppers.png'' color image, where each channel (Red, Green, and Blue) is encoded independently. The CNOT gate \textit{size metric} and computation time for this image are presented in Fig.~\ref{fig:peppers}.
        
        \item Encoding grayscale images from the ``FASHION MNIST'' dataset. We evaluate the performance for composite images with sizes ranging from $10 \times 10$ to $80 \times 80$ images, each image has $28\times 28$ pixels, as shown in Fig.~\ref{fig:fashion_mnist}.
    \end{enumerate}
    
    \begin{figure}
    \begin{subfigure}[htbp]{0.49\textwidth}
        \centering
        \includegraphics[width=1.0\linewidth]{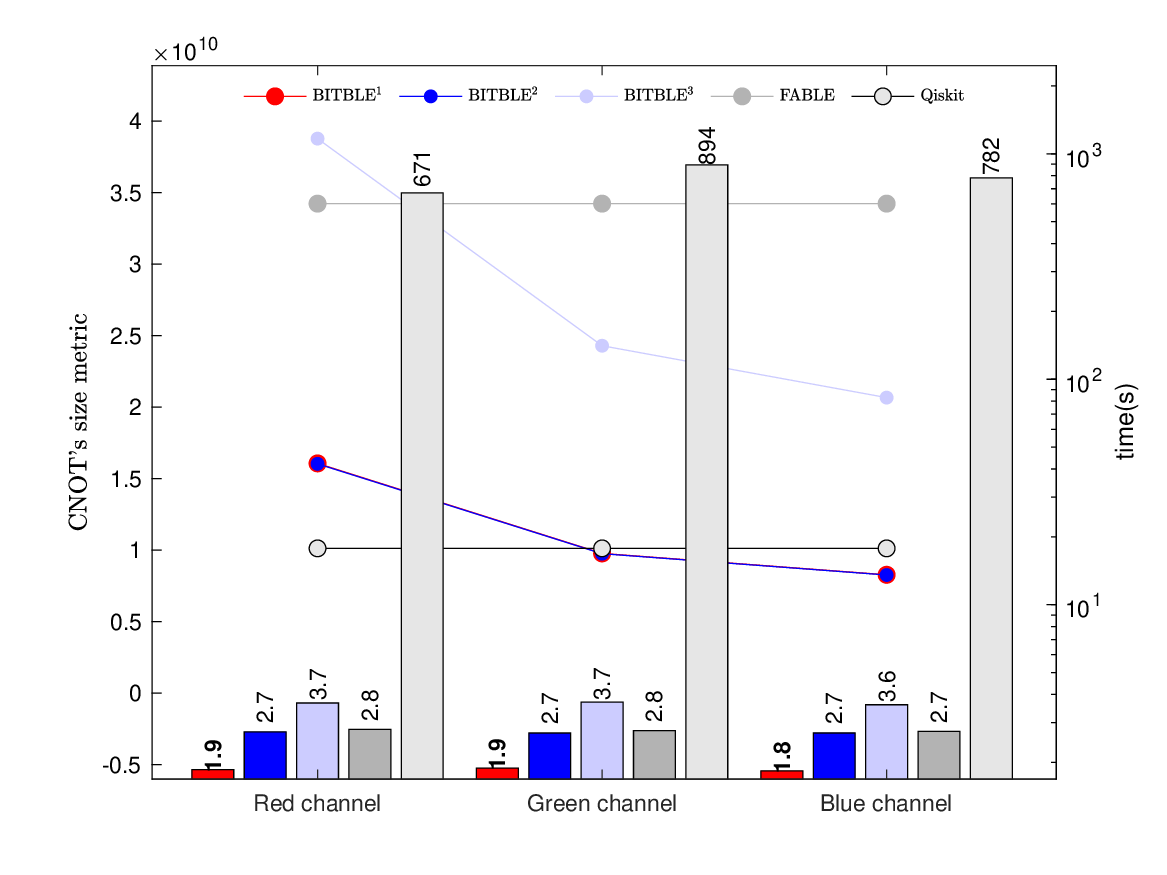}
        \caption{The lines show the \textit{size metric} of CNOT gates on the left $y$-axis and the bars show the time of circuit synthesis on the right $y$-axis. Lower values for both the size metric and time are preferable.}
    \end{subfigure}
    \hfill
    \begin{subfigure}[htbp]{0.49\textwidth}\centering
        \includegraphics[width=1.0\linewidth]{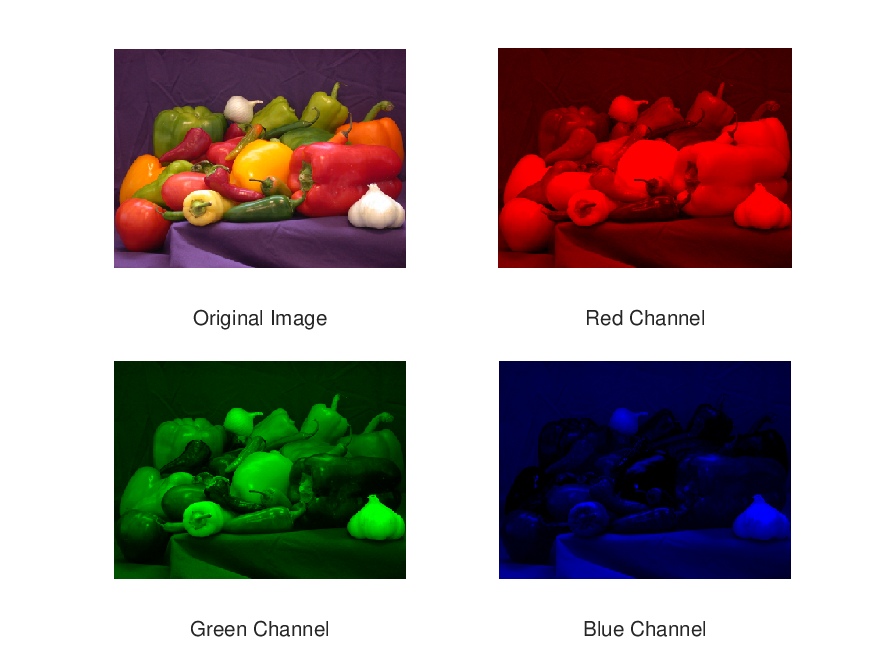}
        \caption{RGB channels of ``Peppers.png''.}
    \end{subfigure}
    \caption{The CNOT's \textit{size metric} and time of computing parameters of RGB channels of ``Peppers.png''.}
    \label{fig:peppers}
    \end{figure}
    \begin{figure}
    \begin{subfigure}[htbp]{0.49\textwidth}
        \centering
        \includegraphics[width=1.0\linewidth]{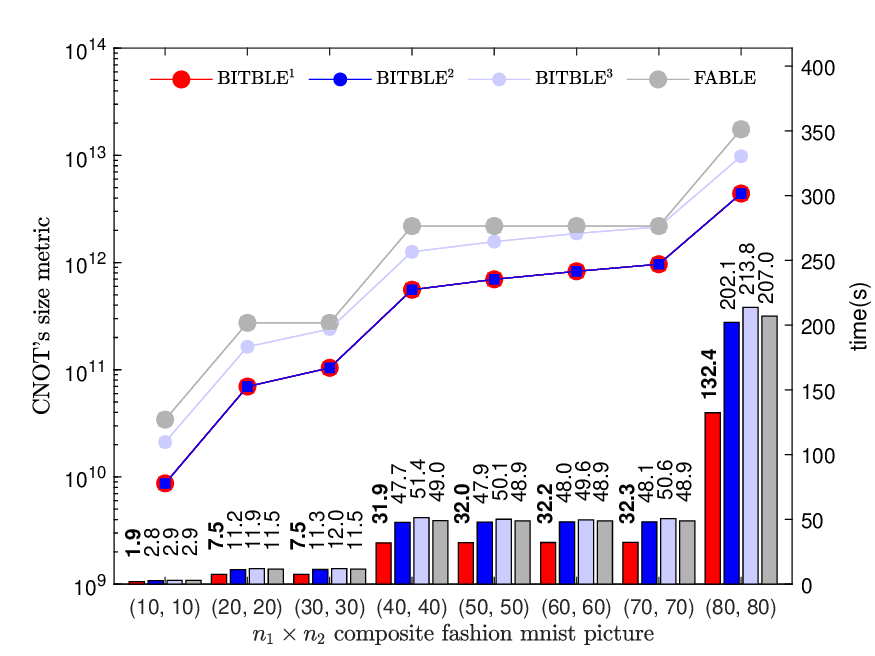}
        \caption{The lines show the \textit{size metric} of CNOT gates on the left $y$-axis and the bars show the time of circuit synthesis on the right $y$-axis. Lower values for both the size metric and time are preferable.}
    \end{subfigure}
    \hfill
    \begin{subfigure}[htbp]{0.49\textwidth}\centering
        \includegraphics[width=1.0\linewidth]{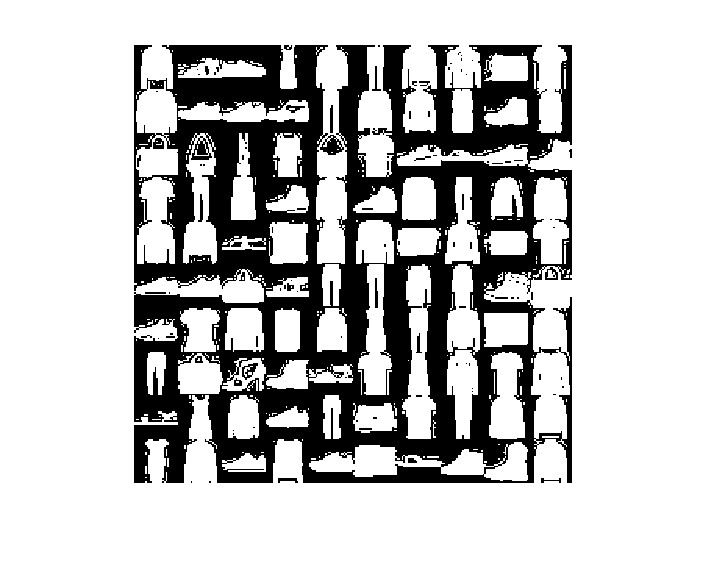}
        \caption{$10\times 10$ composite FASHION MNIST images.}
    \end{subfigure}
    \caption{The CNOT's \textit{size metric} and time of computing parameters of composite ``FASHION MNIST'' images.}
    \label{fig:fashion_mnist}
    \end{figure}
    However, Qiskit's unitary synthesis encounters memory limitations for unitary larger than $1028 \times 1028$ on a 32\,GB RAM computer, and we have not shown the data of Qiskit for encoding ``FASHION MNIST'' images. 
    
    For the ``Peppers.png'' color image encoding, Qiskit's unitary synthesis achieves a lower CNOT \textit{size metric} for the \textit{Red} channel but incurs a $300\times$ longer computation time compared to \texttt{BITBLE} protocols. In contrast, \texttt{BITBLE}\textsuperscript{1} outperforms Qiskit's unitary and \texttt{FABLE} protocol in both time of circuit synthesis and size metric for the \textit{Green} and \textit{Blue} channels.

    For the composite ``Fashion MNIST'' dataset, with pixel sizes ranging from $280 \times 280$ to $2240 \times 2240$ (padding to $512\times 512$ and $4096\times 4096$ matrices), the \texttt{BITBLE}\textsuperscript{1} protocol demonstrates significant advantages over \texttt{FABLE}: it requires $35\%$ less runtime and reduces the CNOT \textit{size metric} by nearly $75\%$. For large-scale datasets, \texttt{BITBLE} supports parallel computation through recursive demultiplexor compared to \texttt{FABLE}. Although this serial implementation does not fully realize the potential speedup, it clearly validates the efficacy of the decoupling approach for block encoding.

\textbf{2) Elliptic partial differential equations.}

    The $1D$ discretized Laplace operator is also considered, which follows the mathematical form:
    $$
    L = 
    \begin{bmatrix}
        2 & -1 & 0 & \cdots & * \\
        -1 & 2 & -1 & \ddots & \vdots \\
        0 & \ddots & \ddots & \ddots & 0 \\
        \vdots & \ddots & -1 & 2 & -1 \\
        * & \cdots & 0 & -1 & 2
    \end{bmatrix},
    $$
    where the $*$ entries in the matrix are both set to $0$ for non-periodic boundary conditions or both set to $-1$ for periodic boundary conditions. In 2D, the discretized Laplace operator becomes the Kronecker sum of discretizations along the $x$- and $y$-directions, a $(n_x, n_y)$-qubit $2D$ discretized Laplace operator~\cite{9951292,osti_1609315} can be defined as
    $$
    L = L_{xx} \oplus L_{yy} = L_{xx} \otimes I_{n_y} + I_{n_x} \otimes L_{yy},
    $$
    where $L_{xx}$ and $L_{yy}$ are $n_x$, $n_y$-qubit operator respectively. The \textit{size metric} of CNOT and single-qubit rotation gates for $1D$ and $2D$ Laplace operators with non-periodic and periodic boundary conditions is presented in Fig.~\ref{fig Elliptic partial differential equations Laplacian}.

\begin{figure}[htbp]
    \centering
    \begin{subfigure}[htbp]{0.23\textwidth}
    \includegraphics[width=1.0\linewidth]{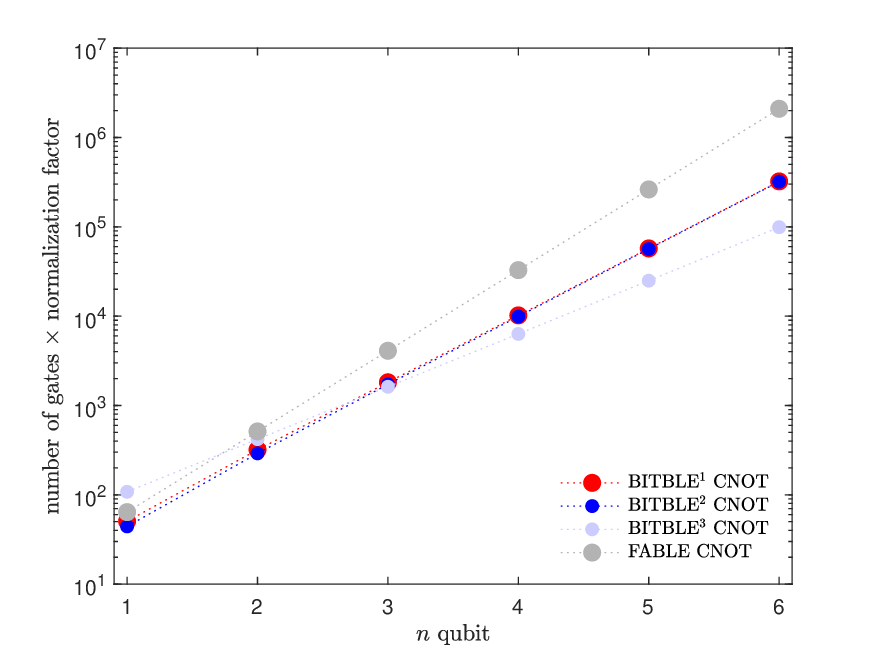}
    \caption{$1D$ non-periodic boundary condition Laplacian}
    \label{fig: BTBLE_FABLE_for_1D_Non_periodic_BC_model}
    \end{subfigure}
    \hfill
    \begin{subfigure}[htbp]{0.23\textwidth}
    \includegraphics[width=1.0\linewidth]{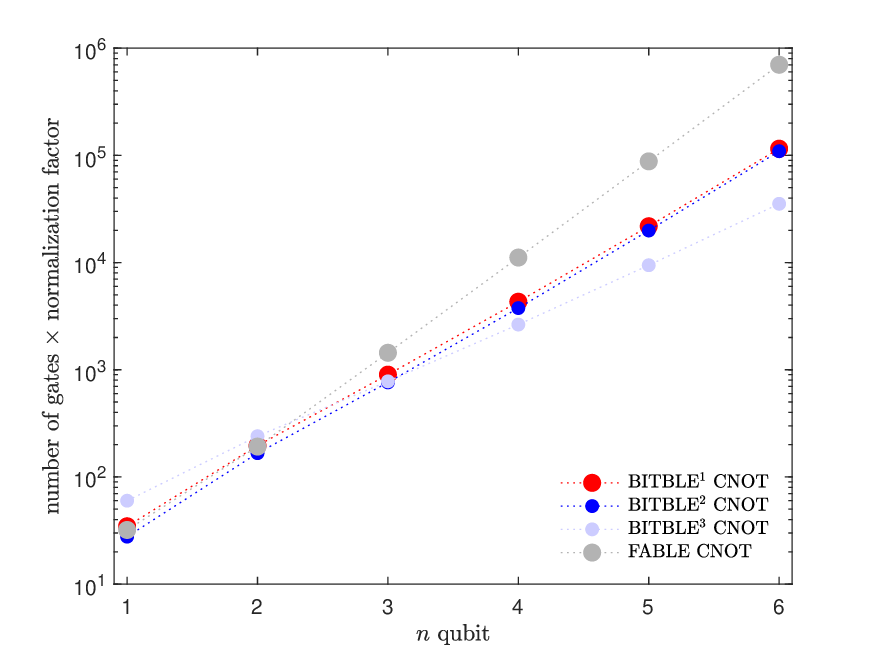}
    \caption{$1D$ periodic boundary condition Laplacian }
    \label{fig: BTBLE_FABLE_for_1D_Periodic_BC_model}
    \end{subfigure}
    \hfill
    \begin{subfigure}[htbp]{0.23\textwidth}
    \includegraphics[width=1.0\linewidth]{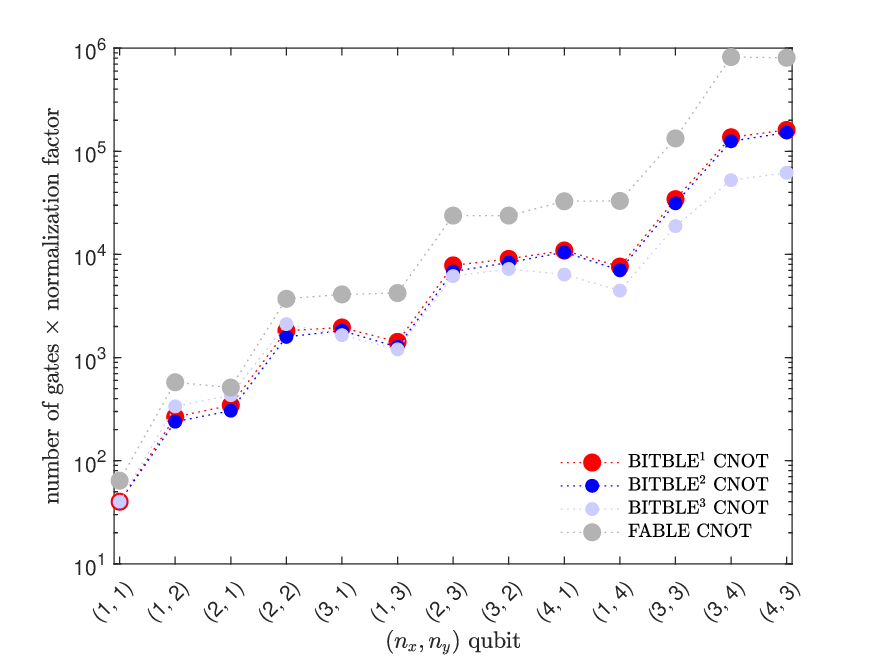}
    \caption{$2D$ non-periodic boundary condition Laplacian}
    \label{fig: BTBLE_FABLE_for_2D_Non_periodic_BC_model}
    \end{subfigure}
    \hfill
    \begin{subfigure}[htbp]{0.23\textwidth}
    \includegraphics[width=1.0\linewidth]{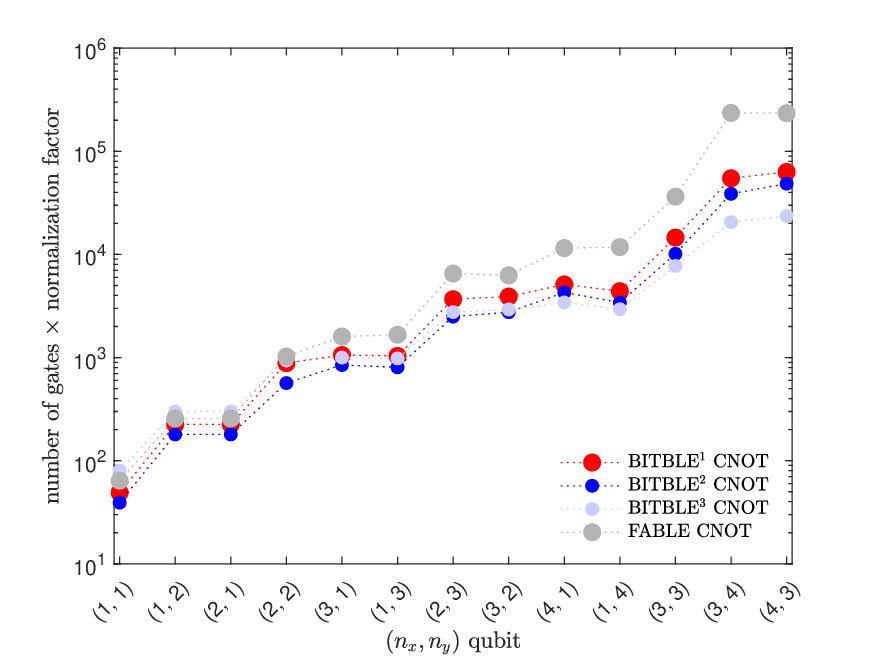}
    \caption{$2D$ periodic boundary condition Laplacian}
    \label{fig: BTBLE_FABLE_for_2D_Periodic_BC_model}
    \end{subfigure}
\caption{The \textit{size metric} of CNOT gates and the rotation-$R_y$ in block-encoding protocols applied to the Laplacian in $1D$ and $2D$ elliptic partial differential equations with non-periodic and periodic boundary conditions. Lower CNOT \textit{size metric} is preferable.}
\label{fig Elliptic partial differential equations Laplacian}
\end{figure}

    The Binary Tree Block-Encoding protocol using normalization factor $\mu_p(A)$ (\texttt{BITBLE}\textsuperscript{3}) exhibits a smaller CNOT \textit{size metric} to encode $1D$ and $2D$ discretized Laplace operators compared to \texttt{BITBLE}\textsuperscript{1}, \texttt{BITBLE}\textsuperscript{2} and \texttt{FABLE}. Specifically, \texttt{BITBLE}\textsuperscript{3} protocol shows more than 90\% reduction in the \textit{size metric} in terms of CNOT gates compared to \texttt{FABLE}~\cite{9951292} in these two systems with more than 4 qubits.

    Based on the above experiments, the \texttt{BITBLE} protocol achieves a better time of circuit synthesis and better \textit{size metric} of CNOT gates compared to the \texttt{FABLE} protocol and Qiskit's unitary synthesis for encoding high-dimensional structured matrices or natural images in most cases.

\subsection{Compression performance}
To compress the circuit size when encoding structured matrices or natural images, the \texttt{BITBLE} protocol optimizes rotation parameters. Its compression function resembles \texttt{FABLE}, eliminating single-qubit rotations with parameters below a threshold and reducible CNOT gates. 

The compression performance of \texttt{BITBLE} is evaluated for Laplace operators in Fig.~\ref{fig compression performance of BITBLE}. The protocol achieves excellent results for:
\begin{itemize}
    \item $1D$ discretized Laplace operators with periodic boundary conditions,
    \item $2D$ discretized Laplace operators,
\end{itemize}
requiring only 10\%--40\% of the maximum gate count for systems with $>6$ qubits. However, its efficiency degrades for $1D$ Laplace operators with \textit{non-periodic} boundary conditions, where it retains over 95\% of the maximum gates.

\begin{figure}[htbp]
    \centering
    \begin{subfigure}[htbp]{0.23\textwidth}
    \includegraphics[width=1.0\linewidth]{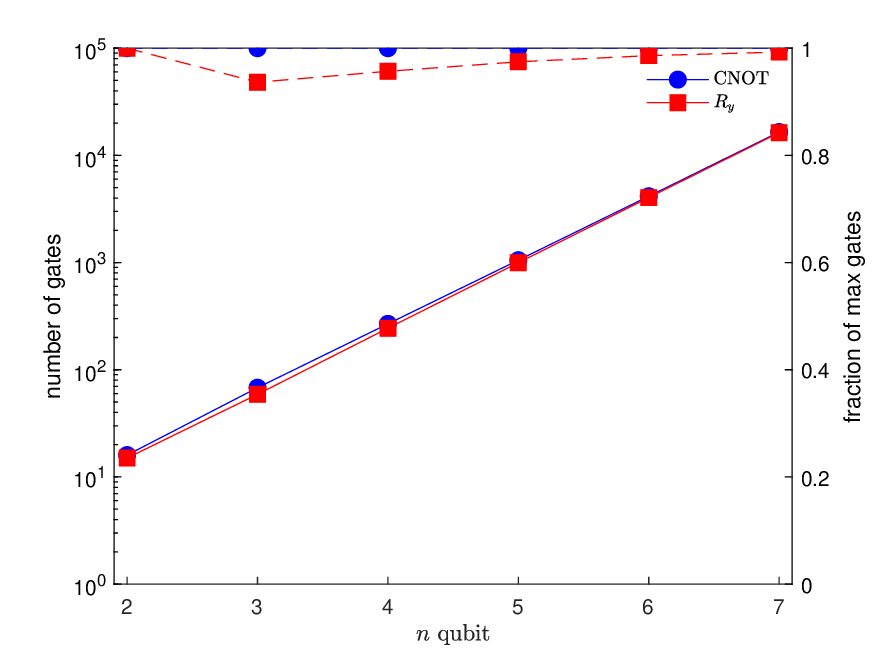}
    \caption{$1D$ non-periodic boundary condition Laplacian}
    \label{fig: BITBLE_for_1D_discretized_Laplace_operators_with_non_periodic_BC}
    \end{subfigure}
    \hfill
    \begin{subfigure}[htbp]{0.23\textwidth}
    \includegraphics[width=1.0\linewidth]{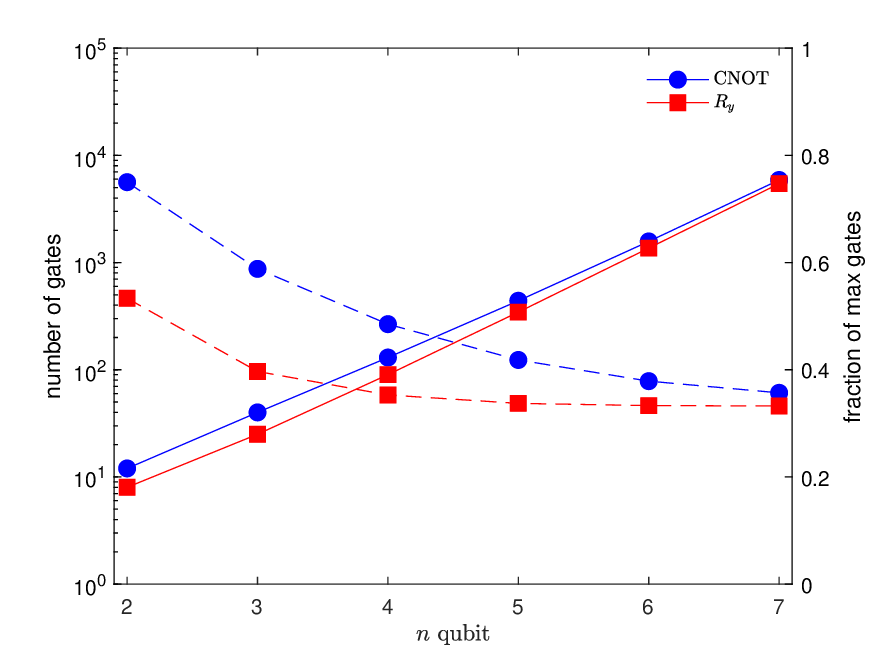}
    \caption{$1D$ periodic boundary condition Laplacian }
    \label{fig: BITBLE_for_1D_discretized_Laplace_operators_with_periodic_BC}
    \end{subfigure}
    \hfill
    \begin{subfigure}[htbp]{0.23\textwidth}
    \includegraphics[width=1.0\linewidth]{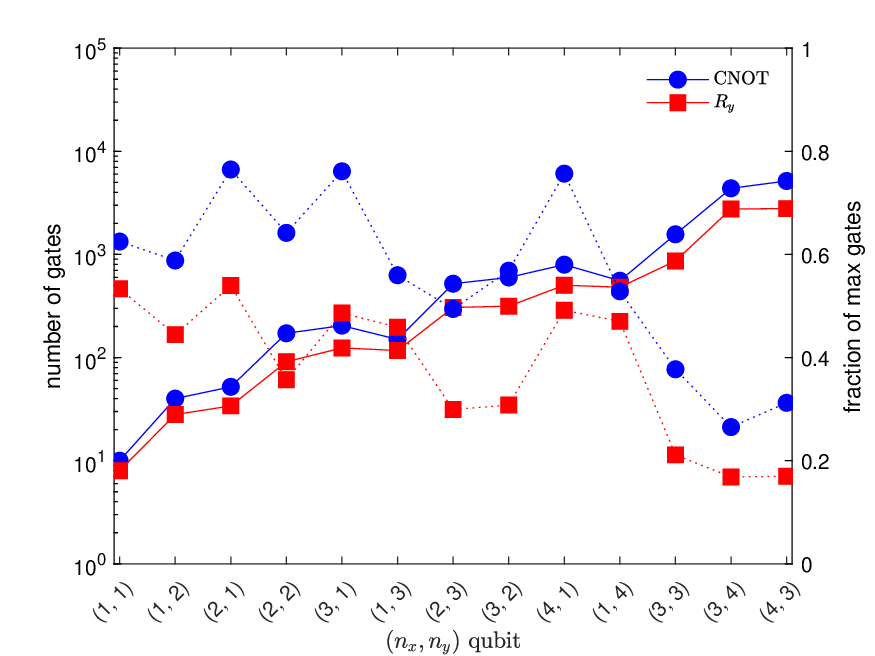}
    \caption{$2D$ non-periodic boundary condition Laplacian}
    \label{fig: BITBLE_for_2D_discretized_Laplace_operators_with_non_periodic_BC}
    \end{subfigure}
    \hfill
    \begin{subfigure}[htbp]{0.23\textwidth}
    \includegraphics[width=1.0\linewidth]{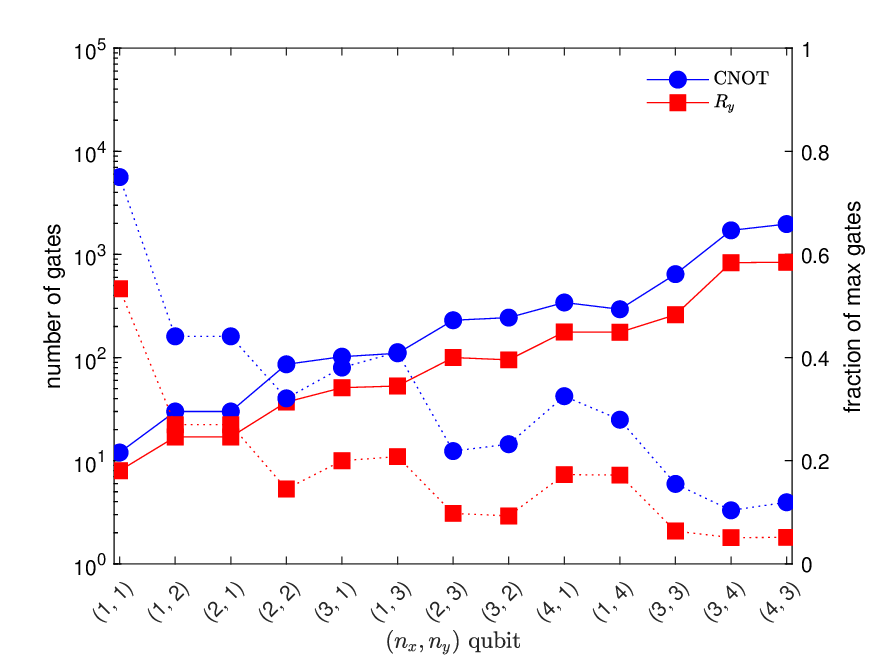}
    \caption{$2D$ periodic boundary condition Laplacian}
    \label{fig: BITBLE_for_2D_discretized_Laplace_operators_with_periodic_BC}
    \end{subfigure}
\caption{Performance compression of \texttt{BITBLE} and \texttt{FABLE}. The full lines show the absolute number of gates on the left $y$-axis and the dotted lines show the fraction of the maximum number of gates on the right $y$-axis. The maximum gates of the \texttt{BITBLE}\textsuperscript{1} protocols in encoding a real matrix of size $2^{n}\times 2^n$ are ($2^{2n}-1$) $R_y$ and ($2^{2n+1}+2^{n+1}-6$) CNOT.}
\label{fig compression performance of BITBLE}
\end{figure}


\section{Conclusion}
\label{section Conclusion}
In this paper, we introduced a block-encoding protocol—Binary Tree Block-Encoding (\texttt{BITBLE})—for generating quantum circuits that block-encode arbitrary target matrices. We proposed two decoupling multiplexor operation methods, permutative demultiplexor and recursive demultiplexor, to reduce the classical computational time of decoupling multiplexor operations, which serve as a subroutine in certain block-encoding protocols. Based on these two decoupling methods, the time of circuit synthesis and normalization factor for encoding a classical matrix have been improved.

We evaluated the effectiveness of \texttt{BITBLE} protocols using the \textit{size metric}---defined as the product of the normalization factor and the gate count---through several example problems, including image processing and encoding of discretized Laplacian operators. Among these protocols, \texttt{BITBLE}\textsuperscript{1} demonstrated superior performance, achieving both the fastest parameter computation time and the lowest size metrics when compared to \texttt{FABLE} and Qiskit's unitary synthesis. This advantage was particularly pronounced for high-dimensional matrices, including structured matrices arising in partial differential equations and natural image encoding.

These block-encoding protocols supported two types of normalization factors: the Frobenius norm $\|A\|_F$ and the $\mu_p(A)$ norm. Our numerical experiments validated the efficiency of parameter computation and the compression performance of \texttt{BITBLE}. Notably, even in serial implementations of the recursive demultiplexer, \texttt{BITBLE} showed promising results, with potential for further speedup through parallel computation.

In the future, one promising direction is the development of parameter computation methods for block-encoding large-scale classical matrices with low circuit depth. Another research direction is exploring the potential speed-up using parallel computation to decouple state-preparation, unitary synthesis, and block-encoding protocols.


\begin{appendices}
\section{Parameters' computing algorithm}
\label{appendix bitble}

The algorithms for computing the parameters in multiplexor operation-$Y$ and multiplexor operation-$Z$ are stated in Algorithm~\ref{alg paramater ry} and Algorithm~\ref{alg paramater rz}, respectively. The complete algorithms for computing multiplexor operation parameters in Binary Tree Block-Encodings (\texttt{BITBLE}\textsuperscript{1} and \texttt{BITBLE}\textsuperscript{2}) with normalization factor $\Vert A\Vert_F$ are stated in Algorithm~\ref{alg parameter bitble}, and Binary Tree Block-Encodings (\texttt{BITBLE}\textsuperscript{3}) with normalization factor $\mu_p(A)$ are stated in Algorithm~\ref{alg parameter bitble3}.

\begin{algorithm}
\caption{angle} 
	\begin{algorithmic}[1]
        \Require Vectors $[a_{j}]_{j=0}^{2^n-1}\in\mathbb{C}^{2^n}$ with $\vert a_j\vert = 1$ for $0\leq j\leq 2^n-1$
        \Ensure Vectors $[\varphi_{j}]_{j=0}^{2^n-1}\in\mathbb{R}^{2^n}$
        \State Get $\varphi_{j}$ such that $ e^{\frac{\varphi_{j}}{2}i}:= \text{Re}(a_{j}) + \text{Im}(a_{j})i$ for all $0\leq j\leq 2^{n}-1$
    \end{algorithmic} 
    \label{alg angles}
\end{algorithm}

\begin{algorithm}
	\caption{R$_\text{Y}$-angles (Compute Rotation-$Y$ binary tree parameters)} 
	\begin{algorithmic}[1]
        \Require The norm of state $[\vert\psi_j\vert]_{j=0}^{2^n-1}$ 
        \Ensure Rotation Parameters $[\varphi_{j}]_{j=0}^{2^n-2}$ 
        \State $a_{n,j}\gets \vert \psi_j\vert$ for all $0\leq j\leq 2^n-1$
		\For {$k = n-1,\ldots,0$}
            \For {$j= 0,\cdots,2^{k}-1$}
                \State $\begin{aligned}
			    &\varphi_{k,j} \gets \text{angle}\left(a_{k+1,2j} + a_{k+1,2j+1}\cdot i\right)
			    \end{aligned}$
                \State $a_{k,j}\gets\sqrt{\vert a_{k+1,2j}\vert^2 + \vert a_{k+1,2j+1}\vert^2}$
            \EndFor
		\EndFor
        \State $\varphi_{2^{k}+j-1}\gets\varphi_{k,j}$ for all $0\leq k\leq n-1$, $0\leq j\leq2^k-1$
	\end{algorithmic} 
    \label{alg paramater ry}
\end{algorithm}

\begin{algorithm}
	\caption{R$_\text{Z}$-angles (Compute Rotation-$Z$ binary tree parameters)} 
	\begin{algorithmic}[1]
        \Require The phases of state $[\phi_k ]_{k=0}^{2^n-1}$ 
        \Ensure Rotation parameters and global phase $\left([\theta_k]_{k=0}^{2^n-2},\theta_{-1}\right)$
        \For{$i = 1,\cdots,n$}
            \For{$0 \leq k\leq 2^{n-i+1}-1$}
                \State $a_k\gets\phi_k$
            \EndFor
            \For{$j = 0,\cdots, 2^{n-i}-1$}
                \State $\phi_j\gets \frac{a_{2j-1} + a_{2j}}{2}$
                \State $\phi_{2^{n-i} + j}\gets -a_{2j-1} + a_{2j}$
            \EndFor
        \EndFor
        \State $\theta_{-1}\gets- a_0 - a_1$ 
        \For{$k = 0,\cdots,2^n-2$}
            \State $\theta_{k}\gets \phi_{k+1}$
        \EndFor
	\end{algorithmic} 
    \label{alg paramater rz}
\end{algorithm}

\begin{algorithm}
	\caption{Compute multiplexor operation parameter in Binary Tree Block-Encodings with normalization factor $\Vert A\Vert_F$ (\texttt{BITBLE\textsuperscript{1}} and \texttt{BITBLE\textsuperscript{2}})} 
	\begin{algorithmic}[1]
        \Require $\left[A_{k,j}:=e^{i\phi_{k,j}}\vert A_{k,j}\vert\right]_{0\leq k,j\leq 2^n-1}$
        \Ensure Compute parameters of the binary tree block-encoding (BITBLE) protocol with normalization factor $\Vert A\Vert_F$ in Fig.~\ref{fig bitble quantum circuit}.
        \For {$j=1,\ldots,2^n$}
        \State $[\varphi_{k,j}]_{k=0}^{2^n-2} \gets \text{R}_\text{Y}\text{-angles}\left([\vert A_{k,j}\vert]_{k=0}^{2^n-1}\right)$
        \If{$A$ is complex}
        \State $\left([\theta_{k,j}]_{k=0}^{2^n-2}, \theta_{-1,j}\right) \gets  \text{R}_\text{Z}\text{-angles}\left([\phi_{k,j}]_{k=0}^{2^n-1}\right)$
        \EndIf
        \EndFor 
        \State $ [\varphi_{j}']_{j=0}^{2^n-2} \gets \text{R}_\text{Y}\text{-angles}\left(\{\Vert A_{\cdot, j}\Vert_2/\Vert A\Vert_F\}_{j=0}^{2^n-1}\right)$
        \State $\hat{\varphi}'_{k,j}\gets\varphi'_{2^k+j-1}$ for all $0\leq k\leq n-1$, $0\leq j\leq2^k-1$ 
	\end{algorithmic} 
    \label{alg parameter bitble}
\end{algorithm}

\begin{algorithm}
	\caption{Compute multiplexor operation parameter in Binary Tree Block-Encodings with normalization factor $\mu_p(A)$ (\texttt{BITBLE$^{3}$})} 
	\begin{algorithmic}[1]
        \Require $\left[A_{k,j}:=e^{i\phi_{k,j}}\vert A_{k,j}\vert\right]_{0\leq k,j\leq 2^n-1}$ and $p\in[0,1]$
        \Ensure Compute parameters of the binary tree block-encoding (BITBLE) protocol with normalization factor $\mu_p(A)$.
        \For {$j=1,\ldots,2^n$}
        \State $[\varphi_{k,j}]_{k=0}^{2^n-2} \gets \text{R}_\text{Y}\text{-angles}\left(\left[\frac{\vert A_{k,j}\vert^p}{\sqrt{\Vert A_{\cdot,j}\Vert_{2p}^{2p}}}\right]_{k=0}^{2^n-1}\right)$
        \State $[\tilde\varphi_{k,j}]_{k=0}^{2^n-2} \gets \text{R}_\text{Y}\text{-angles}\left(\left[\frac{\vert A_{j,k}\vert^{1-p}}{\sqrt{\Vert A_{k,\cdot}\Vert_{2(1-p)}^{2(1-p)}}}\right]_{k=0}^{2^n-1}\right)$
        \If{$A$ is complex}
        \State $\left([\theta_{k,j}]_{k=0}^{2^n-2}, \theta_{-1,j}\right) \gets  \text{R}_\text{Z}\text{-angles}\left([\phi_{k,j}]_{k=0}^{2^n-1}\right)$
        \EndIf
        \EndFor 
        \State $ [\chi_{n,j}^R]_{j=0}^{2^n-1} \gets \text{R}_\text{Y}\text{-angles}\left(\left[\arccos\left(\sqrt{\frac{\Vert A_{\cdot,j}\Vert_{2p}^{2p}}{S_{2p}(A^T)}}\right)\right]_{j=0}^{2^n-1}\right)$
        \State $\begin{aligned}
             &[\chi_{n+1,j}^L]_{j=0}^{2^n-1} \\
             &\quad \gets \text{R}_\text{Y}\text{-angles}\left(\left[\arccos\left(\sqrt{\frac{\Vert A_{k,\cdot}\Vert_{2(1-p)}^{2(1-p)}}{S_{2(1-p)}(A)}}\right)\right]_{k=0}^{2^n-1}\right)
        \end{aligned}$ 
	\end{algorithmic} 
    \label{alg parameter bitble3}
\end{algorithm}

The quantum circuit of Binary Tree Block-Encodings with normalization factor $\mu_p(A)$ (\texttt{BITBLE$^{3}$}) is shown in Fig.~\ref{fig bitble quantum circuit3}.

\begin{figure*}[htbp]
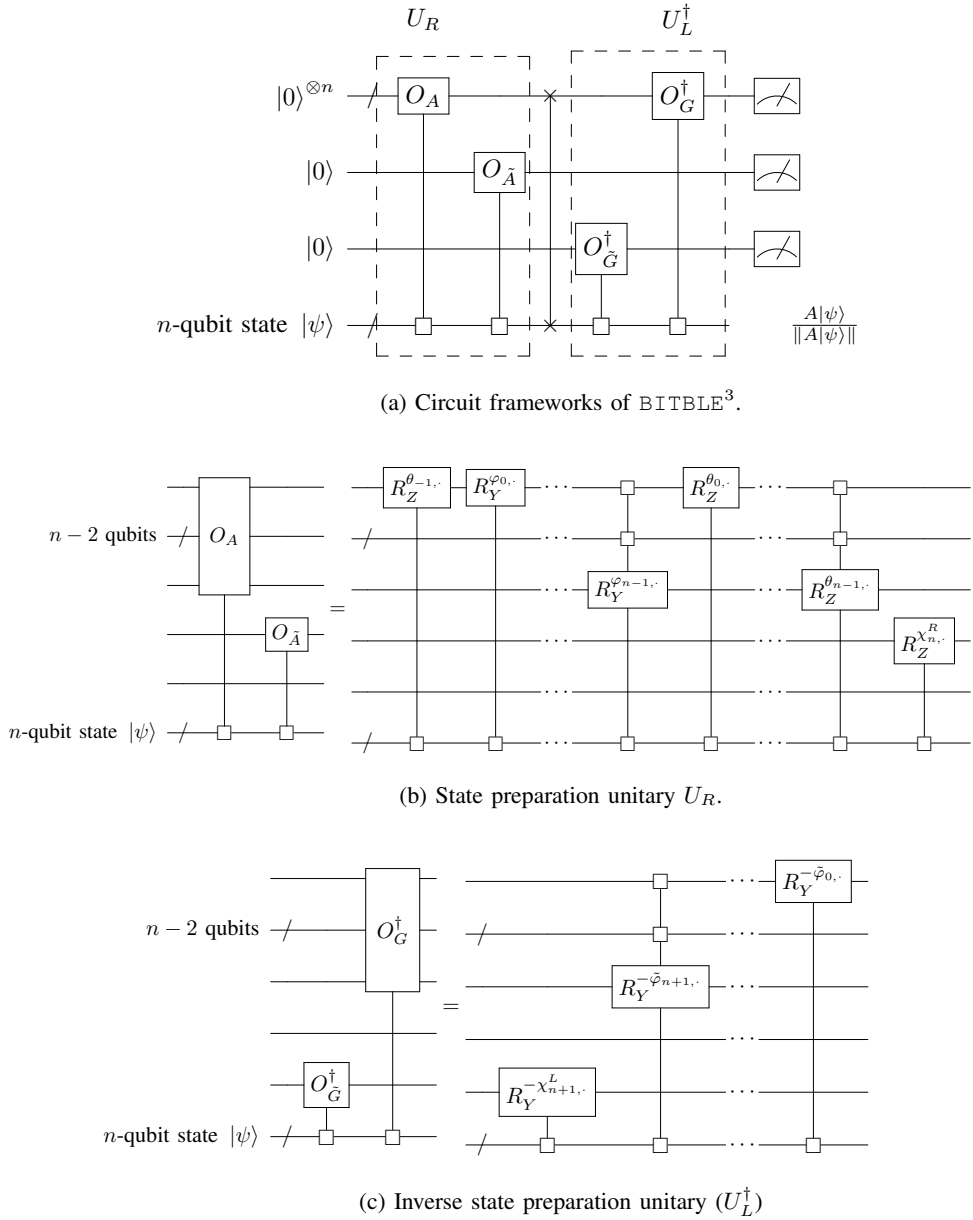

    \begin{subfigure}[htbp]{1.0\textwidth}
    \centering
    \[\resizebox{0.35\textwidth}{!}{$
    \begin{myqcircuit*}{1}{1.0}
		&&& U_{R} &&&& U_{L}^{\dagger} \\
		& \lstick{\ket{0}^{\otimes n}} & {/} \qw & \gate{O_{A}} & \qw & \qswap & \qw & \gate{O_{G}^{\dagger}} & \qw & \meter \\
		& \lstick{\ket{0}} & \qw & \qw & \gate{O_{\tilde{A}}} & \qw & \qw & \qw & \qw & \meter \\
		& \lstick{\ket{0}} & \qw & \qw & \qw & \qw & \gate{O_{\tilde{G}}^{\dagger}} & \qw & \qw & \meter \\
		& \lstick{n\mbox{-qubit state }\ket{\psi}} & {/} \qw & \ctrlsq{-3} & \ctrlsq{-2} & \qswap\qwx[-3] & \ctrlsq{-1} & \ctrlsq{-3} & \qw & \rstick{\frac{A\ket{\psi}}{\left\|A\ket{\psi}\right\|}} 
		\gategroup{2}{4}{5}{5}{1.7em}{--}
		\gategroup{2}{7}{5}{8}{1.7em}{--}
	\end{myqcircuit*}$}\]
    \caption{Circuit frameworks of $\texttt{BITBLE}^3$.}
    \end{subfigure}
    \\[\baselineskip]
    \begin{subfigure}[htbp]{1.0\textwidth}
    \centering
    \[\resizebox{0.6\textwidth}{!}{$
	\begin{myqcircuit*}{0.75}{0.75}
		& & \qw & \multigate{2}{O_{A}} & \qw & \qw \\
		& \lstick{n-2\mbox{ qubits}} & {/} \qw & \ghost{O_{A}} & \qw & \qw \\
		& & \qw & \ghost{O_{A}} & \qw & \qw \\
		& & \qw & \qw & \gate{O_{\tilde{A}}} & \qw \\
		& & \qw & \qw & \qw & \qw \\
		& \lstick{n\mbox{-qubit state }\ket{\psi}} & {/} \qw & \ctrlsq{-3} & \ctrlsq{-2} & \qw 
	\end{myqcircuit*} 
	 = 
	\begin{myqcircuit}
		& \qw & \gate{R_{Z}^{\theta_{-1,\cdot}}} & \gate{R_{Y}^{\varphi_{0,\cdot}}} & \qw & {\cdots} && \ctrlsq{1} & \gate{R_{Z}^{\theta_{0,\cdot}}} & \qw & {\cdots} && \ctrlsq{1} & \qw & \qw \\
		& {/} \qw & \qw & \qw & \qw & {\cdots} && \ctrlsq{1} & \qw & \qw & {\cdots} && \ctrlsq{1} & \qw & \qw \\
		& \qw & \qw & \qw & \qw & {\cdots} && \gate{R_{Y}^{\varphi_{n-1,\cdot}}} & \qw & \qw & {\cdots} && \gate{R_{Z}^{\theta_{n-1,\cdot}}} & \qw & \qw \\
		& \qw & \qw & \qw & \qw & {\cdots} && \qw & \qw & \qw & {\cdots} && \qw & \gate{R_{Z}^{\chi_{n,\cdot}^{R}}} & \qw \\
		& \qw & \qw & \qw & \qw & {\cdots} && \qw & \qw & \qw & {\cdots} && \qw & \qw & \qw \\
		& {/} \qw & \ctrlsq{-5} & \ctrlsq{-5} & \qw & {\cdots} && \ctrlsq{-3} & \ctrlsq{-5} & \qw & {\cdots} && \ctrlsq{-3} & \ctrlsq{-2} & \qw \\
	\end{myqcircuit}	
	$}
	\]
    \caption{State preparation unitary $U_R$.}
    \end{subfigure}
    \\[\baselineskip]
    \begin{subfigure}[htbp]{1.0\textwidth}
    \centering
    \[\resizebox{0.45\textwidth}{!}{$
    \begin{myqcircuit*}{0.3}{0.75}
		& & \qw & \qw & \multigate{2}{O_{G}^{\dagger}} & \qw \\
		& \lstick{n-2\mbox{ qubits}} & {/} \qw & \qw & \ghost{O_{G}^{\dagger}} & \qw \\
		& & \qw & \qw & \ghost{O_{G}^{\dagger}} & \qw \\
		& & \qw & \qw & \qw & \qw \\
		& & \qw & \gate{O_{\tilde{G}}^{\dagger}} & \qw & \qw \\
		& \lstick{n\mbox{-qubit state }\ket{\psi}} & {/} \qw & \ctrlsq{-1} & \ctrlsq{-3} & \qw 
	\end{myqcircuit*} 
	 = 
	\begin{myqcircuit*}{0.2}{0.75}
		& \qw & \qw & \ctrlsq{1} & \qw & {\cdots} && \gate{R_{Y}^{-\tilde{\varphi}_{0,\cdot}}} & \qw \\
		& {/} \qw & \qw & \ctrlsq{1} & \qw & {\cdots} && \qw & \qw \\
		& \qw & \qw & \gate{R_{Y}^{-\tilde{\varphi}_{n+1,\cdot}}} & \qw & {\cdots} && \qw & \qw \\
		& \qw & \qw & \qw & \qw & {\cdots} && \qw & \qw \\
		& \qw & \gate{R_{Y}^{-\chi_{n+1,\cdot}^{L}}} & \qw & \qw & {\cdots} && \qw & \qw \\
		& {/} \qw & \ctrlsq{-1} & \ctrlsq{-3} & \qw & {\cdots} && \ctrlsq{-5} & \qw \\
	\end{myqcircuit*}
    $}
	\]
    \caption{Inverse state preparation unitary ($U_L^\dagger$)}
    \end{subfigure}
\caption{Quantum circuit of Binary Tree Block-encoding with normalization factor $\mu_p(A)$ (\texttt{BITBLE}\textsuperscript{3}).}
\label{fig bitble quantum circuit3}
\end{figure*}

\section{Operation of BITBLE Programme}

All the code is open-course, and the installation process can be found in \href{https://github.com/zexianLIPolyU/BITBLE-SIABLE_matlab}{github}. The following is a demonstration of the code.

\textbf{1) Quantum State Preparation:}
\begin{lstlisting}
%Quantum State Preparation
cd("state_preparation");
addpath('QCLAB');
logging=true; %record
circuit_sim=true;
N=4;
n=log2(N);
state_complex=randn(N,1)+randn(N,1).*1i;
state_comple=state_complex...
./norm(state_complex);
[circuit0,info0]=...
binary_tree_statepreparation(...
state_complex,logging,circuit_sim);
circuit0.draw();
\end{lstlisting}

Output:
$$
\begin{myqcircuit}
\lstick{q_{0}}	&	\gate{R_z}	&	\gate{R_y}	&	\ctrl{1}	&	\qw	&	\ctrl{1}	&	\gate{R_z}	&	\ctrl{1}	&	\qw	&	\ctrl{1}	&	\qw	\\
\lstick{q_{1}}	&	\gate{R_y}	&	\qw	&	\targ	&	\gate{R_y}	&	\targ	&	\gate{R_z}	&	\targ	&	\gate{R_z}	&	\targ	&	\qw	\\
\end{myqcircuit}
$$

\textbf{2) Binary tree block encoding (\texttt{BITBLE}\textsuperscript{1}) using recursive demultiplexor:}
\begin{lstlisting}
% Block-encoding with normalization
% factor$\Vert A\Vert_F$
cd("bitble-qclab");
addpath('QCLAB');
n=1;
A=randn(pow2(n),pow2(n))+...
randn(pow2(n),pow2(n)).*1j;
offset=0;
logging=true;
compr_type='cutoff';
compr_val=1e-8;
circuit_sim=true;
[circuit1,normalization_factor1,info1]...
=bitble(A,compr_type,compr_val,...
logging,offset,circuit_sim);
circuit1.draw();
\end{lstlisting}

\textbf{3) Binary tree block encoding (\texttt{BITBLE}\textsuperscript{2}) using recursive demultiplexor:}
\begin{lstlisting}
% Block-encoding with normalization
% factor$\Vert A\Vert_F$
cd("bitble-qclab");
addpath('QCLAB');
n=1;
A=randn(pow2(n),pow2(n))+...
randn(pow2(n),pow2(n)).*1j;
offset=0;
logging=true;
compr_type='cutoff';
compr_val=1e-8;
circuit_sim=true;
[circuit2,normalization_factor2,info2]...
=bitble2(A,compr_type,compr_val,...
logging,offset,circuit_sim);
circuit2.draw();
\end{lstlisting}

\textbf{4) Binary tree block encoding (\texttt{BITBLE}\textsuperscript{3}) using permutative demultiplexor with normalization factor $\mu_p(A)$}
\begin{lstlisting}
% Block-encoding with normalization
% factor $\mu_p(A)$
cd("bitble-qclab");
addpath('QCLAB');
n=1;
A=randn(pow2(n),pow2(n))+...
randn(pow2(n),pow2(n)).*1j;
offset=0;
logging=true;
compr_type='cutoff';
compr_val=1e-8;
circuit_sim=true;
[circuit3,normalization_factor3,info3]...
=bitble3(A,compr_type,compr_val,...
logging,offset,circuit_sim);
circuit3.draw();
\end{lstlisting}

\begin{figure*}[htbp]
    \begin{subfigure}[htbp]{1.0\textwidth}
    \centering
    \[\begin{myqcircuit*}{1}{1.5}
    \lstick{q_{0}}	&	\gate{R_z}	&	\targ	&	\gate{R_z}	&	\targ	&	\gate{R_y}	&	\targ	&	\gate{R_y}	&	\targ	&	\gate{R_z}	&	\targ	&	\gate{R_z}	&	\targ	&	\qswap\qwx[1]	&	\gate{R_y}	&	\qw  \\
    \lstick{q_{1}}	&	\qw	&	\ctrl{-1}	&	\qw	&	\ctrl{-1}	&	\qw	&	\ctrl{-1}	&	\qw	&	\ctrl{-1}	&	\qw	&	\ctrl{-1}	&	\qw	&	\ctrl{-1}	&	\qswap	&	\qw	&	\qw	\\
    \end{myqcircuit*}\]
    \hfill
    \caption{Output of \texttt{BITBLE}\textsuperscript{1}}
    \end{subfigure}
    \hfill
    \begin{subfigure}[htbp]{1.0\textwidth}
    \centering
    \[\begin{myqcircuit*}{0.2}{0.5}
    \lstick{q_{0}}	&	\gate{R_z}	&	\targ	&	\gate{R_z}	&	\targ	&	\gate{R_y}	&	\targ	&	\gate{R_y}	&	\targ	&	\gate{R_z}	&	\targ	&	\gate{R_z}	&	\targ	&	\qw	&	\qw	&	\qw	&	\qw	&	\qswap\qwx[3]	&	\qw	&	\qw	&	\qw	&	\qw	&	\targ	&	\gate{R_y}	&	\targ	&	\gate{R_y}	&	\qw	 &  \\
    \lstick{q_{1}}	&	\qw	&	\qw	&	\qw	&	\qw	&	\qw	&	\qw	&	\qw	&	\qw	&	\qw	&	\qw	&	\qw	&	\qw	&	\gate{R_y}	&	\targ	&	\gate{R_y}	&	\targ	&	\qw	&	\qw	&	\qw	&	\qw	&	\qw	&	\qw	&	\qw	&	\qw	&	\qw	&	\qw	 &  \\
    \lstick{q_{2}}	&	\qw	&	\qw	&	\qw	&	\qw	&	\qw	&	\qw	&	\qw	&	\qw	&	\qw	&	\qw	&	\qw	&	\qw	&	\qw	&	\qw	&	\qw	&	\qw	&	\qw	&	\targ	&	\gate{R_y}	&	\targ	&	\gate{R_y}	&	\qw	&	\qw	&	\qw	&	\qw	&	\qw	  \\
    \lstick{q_{3}}	&	\qw	&	\ctrl{-3}	&	\qw	&	\ctrl{-3}	&	\qw	&	\ctrl{-3}	&	\qw	&	\ctrl{-3}	&	\qw	&	\ctrl{-3}	&	\qw	&	\ctrl{-3}	&	\qw	&	\ctrl{-2}	&	\qw	&	\ctrl{-2}	&	\qswap	&	\ctrl{-1}	&	\qw	&	\ctrl{-1}	&	\qw	&	\ctrl{-3}	&	\qw	&	\ctrl{-3}	&	\qw	&	\qw	\\
    \end{myqcircuit*}\]
    \hfill
    \caption{Output of \texttt{BITBLE}\textsuperscript{3}}
    \end{subfigure}
    \caption{Output of quantum circuit of three Binary Tree Block Encoding protocols encoding $2\times2$ complex matrix using QCLAB.}
\end{figure*}

\end{appendices}


\section*{Acknowledgment}

\bibliographystyle{IEEEtran}
\bibliography{bibtex/bib/IEEEexample}

\end{document}